\documentclass[fleqn,11pt,twoside]{article}

\pdfoutput=1

\usepackage{amsthm,amsthm,amssymb, color, xcolor,epsfig, graphics, subfigure}

\usepackage{amsmath, graphicx, latexsym, lscape }

\usepackage[english]{babel}
\usepackage{stmaryrd}
\usepackage{mathtools}

\usepackage[breaklinks=true]{hyperref}
\usepackage[numbers]{natbib}

\renewcommand{\d}{\mathrm{d}}

\newcommand{\cE}{\mathcal{E}}
\newcommand{\cF}{\mathcal{F}}
\newcommand{\cL}{\mathcal{L}}
\newcommand{\cV}{\mathcal{V}}
\newcommand{\var}[3]{\frac{\delta_{#1} {#2}}{\delta {#3}}}
\newcommand{\der}[2]{\frac{\partial {#1}}{\partial {#2}}}

\newcommand{\R}{\mathbb{R}}

\newcommand{\tint}{{\textstyle \int}}
\newcommand{\q}[1]{ q^{\scriptscriptstyle [#1]} }

\DeclareMathOperator{\pr}{pr}

\makeatletter
\newcommand{\copyrightnote}[2]{{\renewcommand{\thefootnote}{}
 \footnotetext{\small\it
\begin{flushleft}
 \copyright \ #1   #2  
\end{flushleft}}}}
\def\d{\mbox{\rm d}}

\newcommand{\Name}[1]{\begin{flushleft}
                       \LARGE \bf #1
                       \end{flushleft}\vspace{-3mm}}

\newcommand{\Author}[1]{\begin{flushleft}
                       \it #1 \end{flushleft}}

\newcommand{\Address}[1]{\begin{flushleft}
                       \it #1 \end{flushleft}}

\newcommand{\Date}[1]{\begin{flushleft}
                      \small  \it #1 \end{flushleft}}

\newcommand{\evenhead}{Mats Vermeeren}
\newcommand{\oddhead}{Hamiltonian structures for integrable hierarchies of Lagrangian PDEs}

\renewcommand{\@evenhead}{
\hspace*{-3pt}\raisebox{-15pt}[\headheight][0pt]{\vbox{\hbox to \textwidth
{\thepage \hfil \evenhead}\vskip4pt \hrule}}}
\renewcommand{\@oddhead}{
\hspace*{-3pt}\raisebox{-15pt}[\headheight][0pt]{\vbox{\hbox to \textwidth
{\oddhead \hfil \thepage}\vskip4pt\hrule}}}
\renewcommand{\@evenfoot}{}
\renewcommand{\@oddfoot}{}

\long\def\@makecaption#1#2{%
  \vskip\abovecaptionskip
  \sbox\@tempboxa{\small \textbf{#1.}\ \ #2}%
  \ifdim \wd\@tempboxa >\hsize
    {\small \textbf{#1.}\ \ #2}\par
  \else
    \global \@minipagefalse
    \hb@xt@\hsize{\hfil\box\@tempboxa\hfil}%
  \fi
  \vskip\belowcaptionskip}

\newcommand{\JNMPnumberwithin}[3][\arabic]{%
  \@ifundefined{c@#2}{\@nocounterr{#2}}{%
    \@ifundefined{c@#3}{\@nocnterr{#3}}{%
      \@addtoreset{#2}{#3}%
      \@xp\xdef\csname the#2\endcsname{%
        \@xp\@nx\csname the#3\endcsname .\@nx#1{#2}}}}%
}

\renewenvironment{proof}[1][\proofname]{\par
  \normalfont
  \topsep6\p@\@plus6\p@ \trivlist
  \item[\hskip\labelsep\textbf{%
    #1\@addpunct{.}}]\ignorespaces
}{%
  \qed\endtrivlist
}

\newcommand{\resetfootnoterule} {
  \renewcommand\footnoterule{%
  \kern-3\p@
  \hrule\@width.4\columnwidth
  \kern2.6\p@}
}

\renewcommand{\footnoterule}{}

\makeatother

\newtheorem{thm}{Theorem}[section]
\newtheorem{prop}[thm]{Proposition}

\newtheorem{cor}[thm]{Corollary}
\theoremstyle{definition}
\newtheorem{definition}[thm]{Definition}
\newtheorem{example}[thm]{Example}

\begin{document}

\renewcommand{\evenhead}{ {\LARGE\textcolor{blue!10!black!40!green}{{\sf \ \ \ ]ocnmp[}}}\strut\hfill Mats Vermeeren}
\renewcommand{\oddhead}{ {\LARGE\textcolor{blue!10!black!40!green}{{\sf ]ocnmp[}}}\ \ \ \ \   Hamiltonian structures for hierarchies of Lagrangian PDEs}

%%%% Matter for the first page 
\thispagestyle{empty}
\newcommand{\FistPageHead}[3]{
\begin{flushleft}
\raisebox{8mm}[0pt][0pt]
{\footnotesize \sf
\parbox{150mm}{{Open Communications in Nonlinear Mathematical Physics}\ \  \ {\LARGE\textcolor{blue!10!black!40!green}{]ocnmp[}}
\ \ Vol.1 (2021) pp
#2\hfill {\sc #3}}}\vspace{-13mm}
\end{flushleft}}

\FistPageHead{1}{\pageref{firstpage}--\pageref{lastpage}}{ \ \ Article}

\strut\hfill

\strut\hfill

\copyrightnote{The author(s). Distributed under a Creative Commons Attribution 4.0 International License}

\Name{Hamiltonian structures for integrable hierarchies of Lagrangian PDEs}

\Author{Mats Vermeeren$^{\,1}$}

\Address{$^{1}$ School of Mathematics, University of Leeds, Leeds, LS2 9JT, UK.
	\texttt{m.vermeeren@leeds.ac.uk}}

\Date{Received May 18, 2021; Accepted August 31, 2021}

\setcounter{equation}{0}

\begin{abstract}
\noindent 
Many integrable hierarchies of differential equations allow a variational description, called a Lagrangian multiform or a pluri-Lagrangian structure. The fundamental object in this theory is not a Lagrange function but a differential $d$-form that is integrated over arbitrary $d$-dimensional submanifolds. All such action integrals must be stationary for a field to be a solution to the pluri-Lagrangian problem. In this paper we present a procedure to obtain Hamiltonian structures from the pluri-Lagrangian formulation of an integrable hierarchy of PDEs. As a prelude, we review a similar procedure for integrable ODEs. We show that the exterior derivative of the Lagrangian $d$-form is closely related to the Poisson brackets between the corresponding Hamilton functions. In the ODE (Lagrangian 1-form) case we discuss as examples the Toda hierarchy and the Kepler problem. As examples for the PDE (Lagrangian 2-form) case we present the potential and Schwarzian Korteweg-de Vries hierarchies, as well as the Boussinesq hierarchy.
\end{abstract}

\label{firstpage}

%%%% The Article text starts here

\section{Introduction}

Some of the most powerful descriptions of integrable systems use the Hamiltonian formalism. In mechanics, Liouville-Arnold integrability means having as many independent Hamilton functions as the system has degrees of freedom, such that the Poisson bracket of any two of them vanishes.  In the case of integrable PDEs, which have infinitely many degrees of freedom, integrability is often defined as having an infinite number of commuting Hamiltonian flows, where again each two Hamilton functions have a zero Poisson bracket. In addition, many integrable PDEs have two compatible Poisson brackets. Such a bi-Hamiltonian structure can be used to obtain a recursion operator, which in turn is an effective way to construct an integrable hierarchy of PDEs.

In many cases, especially in mechanics, Hamiltonian systems have an equivalent Lagrangian description. This raises the question whether integrability can be described from a variational perspective too. Indeed, a Lagrangian theory of integrable hierarchies has been developed over the last decade or so, originating in the theory of integrable lattice equations (see for example \cite{lobb2009lagrangian}, \cite{bobenko2010lagrangian}, \cite[Chapter 12]{hietarinta2016discrete}). It is called the theory of \emph{Lagrangian multiform} systems, or, of \emph{pluri-Lagrangian} systems. The continuous version of this theory, i.e.\@ its application to differential equations, was developed among others in \cite{suris2013variational,suris2016lagrangian}. Recently, connections have been established between pluri-Lagrangian systems and variational symmetries \cite{petrera2017variational,petrera2019variational,sleigh2019variational} as well as Lax pairs \cite{sleigh2019lax}. 

Already in one of the earliest studies of continuous pluri-Lagrangian structures \cite{suris2013variational}, the pluri-Lagrangian principle for ODEs was shown to be equivalent to the existence of commuting Hamiltonian flows (see also \cite{sridhar2019commutativity}).  In addition, the property that Hamilton functions are in involution can be expressed in Lagrangian terms as closedness of the Lagrangian form. The main goal of this work is to generalize this connection between pluri-Lagrangian and Hamiltonian structures to the case of integrable PDEs.

A complementary approach to connecting pluri-Lagrangian structures to Hamiltonian structures was recently taken in \cite{caudrelier2020hamiltonian}. There, a generalisation of covariant Hamiltonian field theory is proposed, under the name \emph{Hamiltonian multiform}, as the Hamiltonian counterpart of Lagrangian multiform systems. This yields a Hamiltonian framework where all independent variables are on the same footing. In the present work we obtain classical Hamiltonian structures where one of the independent variables is singled out as the common space variable of all equations in a hierarchy.

We begin this paper with an introduction to pluri-Lagrangian systems in Section \ref{sec-plurilag}. The exposition there relies mostly on examples, while proofs of the main theorems can be found in Appendix \ref{sec-appendix}. Then we discuss how pluri-Lagrangian systems generate Hamiltonian structures, using symplectic forms in configuration space. In Section \ref{sec-1form} we review this for ODEs (Lagrangian 1-form systems) and in Section \ref{sec-2form} we present the case of $(1+1)$-dimensional PDEs (Lagrangian 2-form systems). In each section, we illustrate the results by examples.

\section{Pluri-Lagrangian systems}
\label{sec-plurilag}

A hierarchy of commuting differential equations can be embedded in a higher-dimensional space of independent variables, where each equation has its own time variable. All equations share the same space variables (if any) and have the same configuration manifold $Q$. We use coordinates $t_1,t_2,\ldots,t_N$ in the \emph{multi-time} $M = \R^N$. In the case of a hierarchy of $(1+1)$-dimensional PDEs, the first of these coordinates is a common space coordinate, $t_1 = x$, and we assume that for each $i \geq 2$ there is a PDE in the hierarchy expressing the $t_i$-derivative of a field $u: M \rightarrow Q$ in terms of $u$ and its $x$-derivatives. Then the field $u$ is determined on the whole multi-time $M$ if initial values are prescribed on the $x$-axis. In the case of ODEs, we assume that there is a differential equation for each of the time directions. Then initial conditions at one point in multi-time suffice to determine a solution.

We view a field $u: M \rightarrow Q$ as a smooth section of the trivial bundle $M \times Q$, which has coordinates $(t_1,\ldots,t_N,u)$. The extension of this bundle containing all partial derivatives of $u$ is called the \emph{infinite jet bundle} and denoted by $M \times J^\infty$. Given a field $u$, we call the corresponding section $\llbracket u \rrbracket = (u,u_{t_i},u_{t_it_j},\ldots)$ of the infinite jet bundle the \emph{infinite jet prolongation} of $u$. (See e.g.\@ \cite{anderson1992introduction} or \cite[Sec.\@ 3.5]{olver2000applications}.)

In the pluri-Lagrangian context, the Lagrange function is replaced by a jet-dependent $d$-form. More precisely we consider a fiber-preserving map 
\[ \textstyle \cL: M \times J^\infty \rightarrow \bigwedge^d (T^*M). \]
Since a field $u: M \rightarrow Q$ defines a section of the infinite jet bundle, $\cL$ associates to it a section of $\bigwedge^d (T^*M)$, that is, a $d$-form $\cL\llbracket u \rrbracket$. We use the square brackets $\llbracket u \rrbracket$ to denote dependence on the infinite jet prolongation of $u$. We take $d = 1$ if we are dealing with ODEs and $d = 2$ if we are dealing with $(1+1)$-dimensional PDEs. Higher-dimensional PDEs would correspond to $d > 2$, but are not considered in the present work. (An example of a Lagrangian $3$-form system, the KP hierarchy, can be found in \cite{sleigh2019variational}.) We write 
\[ \cL\llbracket u \rrbracket = \sum_{i} \cL_{i}\llbracket u \rrbracket \, \d t_i \]
for 1-forms and 
\[ \cL\llbracket u \rrbracket = \sum_{i,j} \cL_{ij}\llbracket u \rrbracket \, \d t_i \wedge \d t_j \]
for 2-forms.

\begin{definition}\label{def-pluri}
	A field $u: M \rightarrow Q$ is a solution to the \emph{pluri-Lagrangian problem} for the jet-dependent $d$-form $\cL$, if for every $d$-dimensional submanifold $\Gamma \subset M$ the action $\int_\Gamma \cL\llbracket u \rrbracket $ is critical with respect to variations of the field $u$,
	i.e.\@ 
	\[ \frac{\d}{\d \varepsilon} \int_\Gamma \cL\llbracket u + \epsilon v \rrbracket \bigg|_{\varepsilon = 0} = 0 \]	
	for all variations $v: M \rightarrow Q$ such that $v$ and all its partial derivatives are zero on $\partial \Gamma$.
\end{definition}

Some authors include in the definition that the Lagrangian $d$-form must be closed when evaluated on solutions. That is equivalent to requiring that the action is not just critical on every $d$-submanifold, but also takes the same value on every $d$-submanifold (with the same boundary and topology). In this perspective, one can take variations of the submanifold $\Gamma$ as well as of the fields. We choose not to include the closedness in our definition, because slightly weaker property can be obtained as a consequence Definition \ref{def-pluri} (see Proposition \ref{prop-almost-closed} in the Appendix). Most of the authors that include closedness in the definition use the term ``Lagrangian multiform'' (e.g.\@ \cite{lobb2009lagrangian,hietarinta2016discrete,xenitidis2011lagrangian,yoo2011discrete, sleigh2019variational}), whereas those that do not tend to use ``pluri-Lagrangian'' (e.g.\@ \cite{bobenko2015discrete,boll2014integrability,suris2016variational}). Whether or not it is included in the definition, closedness of the Lagrangian $d$-form is an important property. As we will see in Sections \ref{sec-closed1} and \ref{sec-closed2}, it is the Lagrangian counterpart to vanishing Poisson brackets between Hamilton functions.

Clearly the pluri-Lagrangian principle is stronger than the usual variational principle for the individual coefficients $\cL_i$ or $\cL_{ij}$ of the Lagrangian form. Hence the classical Euler-Lagrange equations are only a part of the system equations characterizing a solution to the pluri-Lagrangian problem. This system, which we call the \emph{multi-time Euler-Lagrange equations}, was derived in \cite{suris2016lagrangian} for Lagrangian 1- and 2-forms by approximating an arbitrary given curve or surface $\Gamma$ by \emph{stepped} curves or surfaces, which are piecewise flat with tangent spaces spanned by coordinate directions. In Appendix \ref{sec-appendix} we give a more intrinsic proof that the multi-time Euler-Lagrange equations imply criticality in the pluri-Lagrangian sense. Yet another proof can be found in \cite{sleigh2020lagrangian}.

In order to write the multi-time Euler-Lagrange equations in a convenient form, we will use the multi-index notation for (mixed) partial derivatives. Let $I$ be an $N$-index, i.e.\@ a $N$-tuple of non-negative integers. We denote by $u_I$ the mixed partial derivative of $u: \R^N \rightarrow Q$, where the number of derivatives with respect to each $t_i$ is given by the entries of $I$. Note that if $I = (0,\ldots,0)$, then $u_I = u$. We will often denote a multi-index suggestively by a string of $t_i$-variables, but it should be noted that this representation is not always unique. For example,
\[ t_1 = (1,0,\ldots,0), \qquad t_N = (0,\ldots,0,1), \qquad t_1 t_2 = t_2 t_1 = (1,1,0,\ldots,0) . \]
In this notation we will also make use of exponents to compactify the expressions, for example
\[ t_2^3 = t_2 t_2 t_2 = (0,3,0,\ldots,0). \]
The notation $I t_j$ should be interpreted as concatenation in the string representation, hence it denotes the multi-index obtained from $I$ by increasing the $j$-th entry by one. Finally, if the $j$-th entry of $I$ is nonzero we say that $I$ contains $t_j$, and write $I \ni t_j$.

\subsection{Lagrangian 1-forms}

\begin{thm}[\cite{suris2016lagrangian}]\label{thm-EL1}
	Consider the Lagrangian 1-form 
	\[ \cL\llbracket u \rrbracket = \sum_{j=1}^N \cL_j\llbracket u \rrbracket \,\d t_j, \]
	depending on an arbitrary number of derivatives of $u$. A field $u$ is critical in the pluri-Lagrangian sense if and only if it satisfies the multi-time Euler-Lagrange equations
	\begin{align}
		&\label{el1} \var{j}{\cL_j}{u_I} = 0 & \forall I \not\ni t_j , \\
		&\label{el2} \var{j}{\cL_j}{u_{It_j}} - \var{1}{\cL_1}{u_{It_1}} = 0 &  \forall I,
	\end{align}
	for all indices $j \in \{1,\ldots,N\}$, where $\var{j}{}{u_I}$ denotes the variational derivative in the direction of $t_j$ with respect to $u_I$, 
	\[ \var{j}{}{u_I} = \der{}{u_I} - \partial_j \der{}{u_{It_j}} + \partial_j^2 \der{}{u_{It_jt_j}} - \cdots , \]
	and $\partial_j = \frac{\d}{\d t_j}$.
\end{thm}

Note the derivative $\partial_j$ equals the total derivative $\sum_I u_{It_j} \der{}{u_I}$ if it is applied to a function $f\llbracket u \rrbracket$ that only depends on $t_j$ through $u$. Using the total derivative has the advantage that calculations can be done on an algebraic level,  where the $u_I$ are formal symbols that do not necessarily have an analytic interpretation as a derivative.

\begin{example}\label{ex-toda}
	The Toda lattice describes $N$ particles on a line with an exponential nearest-neighbor interaction. We denote the displacement from equilibrium of the particles by $u = (\q{1},\ldots,\q{N})$. We impose either periodic boundary conditions (formally $\q{0} = \q{N}$ and $\q{N+1} = \q{1}$) or open-ended boundary conditions (formally $\q{0} = \infty$ and $\q{N+1} = -\infty$). We will use $\q{k}_j$ as shorthand notation for the derivative $\q{k}_{t_j} = \frac{\d \q{k}}{\d t_j}$. Consider the hierarchy consisting of the Newtonian equation for the Toda lattice,
	\begin{equation}\label{toda-equation}
		\q{k}_{11} = \exp\!\left( \q{k+1} - \q{k} \right) - \exp\!\left( \q{k} - \q{k-1} \right) , \\
	\end{equation}
	along with its variational symmetries,
	\begin{equation}\label{toda-hierarchy}
		\begin{split}
			\q{k}_2 &= \left( \q{k}_1 \right)^2 + \exp\!\left( \q{k+1} - \q{k} \right) + \exp\!\left( \q{k} - \q{k-1} \right) ,\\
			\q{k}_3 &= \left( \q{k}_1 \right)^3 + \q{k+1}_1 \exp\!\left( \q{k+1} - \q{k} \right) + \q{k-1}_1\exp\!\left( \q{k} - \q{k-1} \right) \\
			&\quad + 2 \q{k}_1 \left( \exp\!\left( \q{k+1} - \q{k} \right) + \exp\!\left( \q{k} - \q{k-1} \right) \right),\\
			&\vdotswithin{=}
		\end{split}
	\end{equation}
	The hierarchy \eqref{toda-equation}--\eqref{toda-hierarchy} has a Lagrangian 1-form with coefficients
	\begin{align*}
		\cL_1 &= \sum_k \left( \frac{1}{2} \left( \q{k}_1 \right)^2 - \exp\!\left( \q{k} - \q{k-1} \right) \right) , \\
		\cL_2 &= \sum_k \left( \q{k}_1 \q{k}_2 - \frac{1}{3}\left( \q{k}_1 \right)^3 - \left( \q{k}_1 + \q{k-1}_1 \right) \exp\!\left( \q{k} - \q{k-1} \right) \right) , \\
		\cL_3 &= \sum_k \bigg( -\frac{1}{4} \left( \q{k}_1 \right)^{4} - \left( \left( \q{k+1}_1 \right)^2 + \q{k+1}_1 \q{k}_1 + \left( \q{k}_1 \right)^2 \right) \exp\!\left( \q{k+1} - \q{k} \right) \\
		&\hspace{1.5cm} + \q{k}_1 \q{k}_3 - \exp\!\left( \q{k+2} - \q{k} \right) - \frac{1}{2} \exp\!\left( 2 (\q{k+1} - \q{k}) \right) \bigg) ,\\
		&\vdotswithin{=}
	\end{align*}
	See \cite{petrera2017variational,vermeeren2019continuum} for constructions of this pluri-Lagrangian structure. The classical Euler-Lagrange equations of these Lagrangian coefficients are
	\begin{align*}
		\var{1}{\cL_1}{\q{k}} = 0 \quad &\Leftrightarrow \quad
		\q{k}_{11} = e^{\q{k+1} - \q{k}} - e^{\q{k} - \q{k-1}} , \\
		\var{2}{\cL_2}{\q{k}} = 0 \quad &\Leftrightarrow \quad
		\q{k}_{12} = 
		\left( \q{k}_1 + \q{k+1}_1 \right) e^{\q{k+1} - \q{k}} - \left( \q{k-1}_1 + \q{k}_1 \right) e^{\q{k} - \q{k-1}} , \\
		&\vdotswithin{\Rightarrow}
	\end{align*}
	We recover Equation \eqref{toda-equation}, but for the other equations of the hierarchy we only get a differentiated form. However, we do get their evolutionary form, as in Equation  \eqref{toda-hierarchy}, from the multi-time Euler-Lagrange equations 
	\[ \var{2}{\cL_2}{\q{k}_1} = 0, \qquad \var{3}{\cL_3}{\q{k}_1} = 0, \qquad \cdots . \]
	The multi-time Euler-Lagrange equations of type \eqref{el2} are trivially satisfied in this case:
	$\var{j}{\cL_j}{\q{k}_j} = \q{k}_1$ for all $j$.
\end{example}

\subsection{Lagrangian 2-forms}

\begin{thm}[\cite{suris2016lagrangian}]\label{thm-EL2}
	Consider the Lagrangian 2-form 
	\[ \cL\llbracket u \rrbracket = \sum_{i < j} \cL_{ij}\llbracket u \rrbracket \,\d t_i \wedge \d t_j, \]
	depending on an arbitrary number of derivatives of $u$.	A field $u$ is critical in the pluri-Lagrangian sense if and only if it satisfies the multi-time Euler-Lagrange equations
	\begin{align}
		&\label{2el1} \var{ij}{\cL_{ij}}{u_I} = 0 & \forall I \not\ni t_i,t_j , \\
		&\label{2el2} \var{ij}{\cL_{ij}}{u_{It_j}} - \var{ik}{\cL_{ik}}{u_{It_k}} = 0 & \forall I \not\ni t_i , \\
		&\label{2el3} \var{ij}{\cL_{ij}}{u_{It_it_j}} + \var{jk}{\cL_{jk}}{u_{It_jt_k}} + \var{ki}{\cL_{ki}}{u_{It_kt_i}}= 0 & \forall I ,
	\end{align}
	for all triples $(i,j,k)$ of distinct indices, where
	\[ \var{ij}{}{u_I} = \sum_{\alpha,\beta = 0}^\infty (-1)^{\alpha+\beta} \partial_i^\alpha \partial_j^\beta \der{}{u_{I t_i^\alpha t_j^\beta}} . \]
\end{thm}

\begin{example} 
	\label{ex-kdv}
	A Lagrangian 2-form for the potential KdV hierarchy was first given in \cite{suris2016lagrangian}. It is instructive to look at just two of the equations embedded in $\R^3$. Then the Lagrangian 2-form has three coefficients,
	\[ \cL = \cL_{12} \,\d t_1 \wedge \d t_2 +  \cL_{13} \,\d t_1 \wedge \d t_3 +  \cL_{23} \,\d t_2 \wedge \d t_3, \]
	where $t_1$ is viewed as the space variable.
	We can take 
	\begin{align*}
		\cL_{12} & =  -u_{1}^3 - \frac{1}{2} u_{1} u_{111} + \frac{1}{2} u_{1} u_{2} , \\
		\cL_{13} & = -\frac{5}{2} u_{1}^4 + 5 u_{1} u_{11}^2 - \frac{1}{2} u_{111}^2 + \frac{1}{2} u_{1} u_{3} ,
	\end{align*}
	where $u_i$ is a shorthand notation for the partial derivative $u_{t_i}$, and similar notations are used for higher derivatives. These are the classical Lagrangians of the potential KdV hierarchy. However, their classical Euler-Lagrange equations give the hierarchy only in a differentiated form,
	\begin{align*}
		u_{12} &= 6 u_1 u_{11} + u_{1111}, \\
		u_{13} &= 30 u_1^2 u_{11} + 20 u_{11} u_{111} + 10 u_1 u_{1111} + u_{111111}.
	\end{align*}
	The Lagrangian 2-form also contains a coefficient 
	\begin{align*}
		\cL_{23} & = 3 u_{1}^5 - \frac{15}{2} u_{1}^2 u_{11}^2 + 10 u_{1}^3 u_{111} - 5 u_{1}^3 u_{2} + \frac{7}{2} u_{11}^2 u_{111} + 3 u_{1} u_{111}^2 - 6 u_{1} u_{11} u_{1111}  \\
		&\quad + \frac{3}{2} u_{1}^2 u_{11111} + 10 u_{1} u_{11} u_{12} - \frac{5}{2} u_{11}^2 u_{2} - 5 u_{1} u_{111} u_{2} + \frac{3}{2} u_{1}^2 u_{3} - \frac{1}{2} u_{1111}^2 \\
		&\quad + \frac{1}{2} u_{111} u_{11111} - u_{111} u_{112} + \frac{1}{2} u_{1} u_{113} + u_{1111} u_{12} - \frac{1}{2} u_{11} u_{13} - \frac{1}{2} u_{11111} u_{2} \\
		&\quad + \frac{1}{2} u_{111} u_{3}
	\end{align*}
	which does not have a classical interpretation, but contributes meaningfully in the pluri-Lagrangian formalism. In particular, the multi-time Euler-Lagrange equations
	\[ \var{12}{\cL_{12}}{u_1} + \var{23}{\cL_{23}}{u_3} = 0 \qquad \text{and} \qquad \var{13}{\cL_{13}}{u_1} - \var{23}{\cL_{23}}{u_3} = 0\]
	yield the potential KdV equations in their evolutionary form,
	\begin{align*}
		u_2 &= 3 u_1^2 + u_{111}, \\
		u_3 &= 10 u_1^3 + 5 u_{11}^2 + 10 u_1 u_{111} + u_{11111}.
	\end{align*}
	All other multi-time Euler-Lagrange equations are consequences of these evolutionary equations.
	
	This example can be extended to contain an arbitrary number of equations from the potential KdV hierarchy. The coefficients $\cL_{1j}$ will be Lagrangians of the individual equations, whereas the $\cL_{ij}$ for $i,j>1$ do not appear in the traditional Lagrangian picture.
\end{example}

\begin{example}\label{ex-bsq}
	The Boussinesq equation
	\begin{equation}\label{bsq}
		u_{22} = 12 u_{1} u_{11} - 3 u_{1111}
	\end{equation}
	is of second order in its time $t_2$, but the higher equations of its hierarchy are of first order in their respective times, beginning with
	\begin{equation}\label{bsq3}
		u_{3} = -6 u_{1} u_{2} + 3 u_{112} .
	\end{equation}
	A Lagrangian 2-form for this system has coefficients
	\begin{align*}
		\cL_{12} &= \frac{1}{2} u_{2}^2 - 2 u_{1}^3 - \frac{3}{2} u_{11}^2 , \\
		\cL_{13} &= u_{2} u_{3} + 6 u_{1}^4 + 27 u_{1} u_{11}^2 - 6 u_{} u_{12} u_{2} + \frac{9}{2} u_{111}^2 + \frac{3}{2} u_{12}^2 , \\
		\cL_{23} &= 24 u_{1}^3 u_{2} + 18 u_{1} u_{11} u_{12} + 9 u_{11}^2 u_{2} - 18 u_{1} u_{111} u_{2} - 2 u_{2}^3 - 6 u_{} u_{2} u_{22} \\
		&\quad + 6 u_{1}^2 u_{3} + 9 u_{111} u_{112} + 3 u_{11} u_{13} + 3 u_{12} u_{22} - 3 u_{111} u_{3} .
	\end{align*}
	They can be found in \cite{vermeeren2019variational} with a different scaling of $\cL$ and a different numbering of the time variables. Equation \eqref{bsq} is equivalent to the Euler-Lagrange equation
	\[ \var{12}{\cL_{12}}{u} = 0 \]
	and Equation \eqref{bsq3} to
	\[ \var{13}{\cL_{13}}{u_2} = 0 . \]
	All other multi-time Euler-Lagrange equations are differential consequences of Equations \eqref{bsq} and \eqref{bsq3}. As in the previous example, it is possible to extend this 2-form to represent an arbitrary number of equations from the hierarchy.
\end{example}

Further examples of pluri-Lagrangian 2-form systems can be found in \cite{sleigh2019lax, sleigh2019variational, vermeeren2019continuum, vermeeren2019variational}.

\section{Hamiltonian structure of Lagrangian 1-form systems}
\label{sec-1form}

A connection between Lagrangian 1-form systems and Hamiltonian or symplectic systems was found in \cite{suris2013variational}, both in the continuous and the discrete case. Here we specialize that result to the common case where one coefficient of the Lagrangian 1-form is a mechanical Lagrangian and all others are linear in their respective time-derivatives. We formulate explicitly the underlying symplectic structures, which will provide guidance for the case of Lagrangian 2-form systems. Since some of the coefficients of the Lagrangian form will be linear in velocities, it is helpful to first have a look at the Hamiltonian formulation for Lagrangians of this type, independent of a pluri-Lagrangian structure.

\subsection{Lagrangians that are linear in velocities}
Let the configuration space be a finite-dimensional real vector space $Q = \R^N$ and consider a Lagrangian $\cL: T Q \rightarrow \R$ of the form
\begin{equation}\label{linlag}
	\cL(q,q_t) = p(q)^T q_t - V(q) ,
\end{equation}
where
\begin{equation}\label{degeneracy}
	\det \left( \der{p}{q} - \left(\der{p}{q}\right)^T \right) \neq 0 .
\end{equation}
Note that $p$ denotes a function of the position $q$; later on we will use $\pi$ to denote the momentum as an element of cotangent space. If $Q$ is a manifold, the arguments of this subsection will still apply if there exists local coordinates in which the Lagrangian is of the form \eqref{linlag}. The Euler-Lagrange equations are first order ODEs:
\begin{equation}\label{EL1}
	\dot{q} =  \left( \left(\der{p}{q}\right)^T  -  \der{p}{q} \right)^{-1} \nabla V ,
\end{equation}
where $\nabla V = \left( \der{V}{q} \right)^T$ is the gradient of $V$.

Note that Equation \eqref{degeneracy} implies that $Q$ is even-dimensional, hence $Q$ admits a (local) symplectic structure. Instead of a symplectic form on $T^*Q$, the Lagrangian system preserves a symplectic form on $Q$ itself \cite{bergvelt1985poisson,rowley2002variational}:
\begin{align}\label{symp}
	\omega = \sum_i - \d p_i(q) \wedge \d q_i 
	&= \sum_{i,j} - \der{p_i}{q_j} \,\d q_j \wedge \d q_i \\
	&= \sum_{i < j} \left(\der{p_i}{q_j} - \der{p_j}{q_i} \right) \d q_i \wedge \d q_j , \notag
\end{align}
which is non-degenerate by virtue of Equation \eqref{degeneracy}.

\begin{prop}\label{prop-linlag}
	The Euler-Lagrange equation \eqref{EL1} of the Lagrangian \eqref{linlag} corresponds to a Hamiltonian vector field with respect to the symplectic structure $\omega$, with Hamilton function $V$.
\end{prop}
\begin{proof}
	The Hamiltonian vector field $X = \sum_i X_i \der{}{q_i}$ of the Hamilton function $V$ with respect to $\omega$ satisfies
	\[ \iota_X \omega = \d V, \]
	where 
	\[ \iota_X \omega = \sum_i \sum_{j \neq i} \left( \der{p_j}{q_i} - \der{p_i}{q_j} \right) X_j \, \d q_i \]
	and 
	\[ \d V = \sum_i \der{V}{q_i} \, \d q_i . \]
	Hence
	\[ X = \left( \left(\der{p}{q}\right)^T  -  \der{p}{q} \right)^{-1} \nabla V ,\]
	which is the vector field corresponding to the Euler-Lagrange equation \eqref{EL1}.
\end{proof}

\subsection{From pluri-Lagrangian to Hamiltonian systems}
\label{sec-to-ham1} 

On a finite-dimensional real vector space $Q$, consider a Lagrangian 1-form $\cL = \sum_i \cL_i \,\d t_i$ consisting of a mechanical Lagrangian 
\begin{equation}\label{L-1-1}
	\cL_1(q,q_1) = \frac{1}{2} |q_1|^2 - V_1(q),
\end{equation}
where $|q_1|^2 = q_1^T q_1$, and additional coefficients of the form
\begin{equation}\label{L-1-2}
	\cL_i(q,q_1,q_i) = q_1^T q_i - V_i(q,q_1) \qquad \text{for } i \geq 2 ,
\end{equation}
where the indices of $q$ denote partial derivatives, $q_i = q_{t_i} = \frac{\d q}{\d t_i}$, whereas the indices of $\cL$ and $V$ are labels. We have chosen the Lagrangian coefficients such that they share a common momentum $p = q_1$, which is forced upon us by the multi-time Euler-Lagrange equation \eqref{el2}. Note that for each $i$, the coefficient $\cL_i$ contains derivatives of $q$ with respect to $t_1$ and $t_i$ only. Many Lagrangian 1-forms are of this form, including the Toda hierarchy, presented in  Example \ref{ex-toda}.

The nontrivial multi-time Euler-Lagrange equations are
\[ \var{1}{\cL_1}{q} = 0 \quad \Leftrightarrow \quad q_{11} = - \der{V_1}{q}, \]
and
\[ \var{i}{\cL_i}{q_1} = 0 \quad \Leftrightarrow \quad q_{i} = \der{V_i}{q_1} \qquad \qquad \text{for } i \geq 2, \]
with the additional condition that
\[\var{i}{\cL_i}{q} = 0 \quad \Leftrightarrow \quad q_{1i} + \der{V_i}{q} = 0.\]
Hence the multi-time Euler-Lagrange equations are overdetermined. Only for particular choices of $V_i$ will the last equation be a differential consequence of the other multi-time Euler-Lagrange equations. The existence of suitable $V_i$ for a given hierarchy could be taken as a definition of its integrability.

Note that there is no multi-time Euler-Lagrange equation involving the variational derivative
\[\var{1}{\cL_i}{q} = \der{V_i}{q} - \frac{\d}{\d t_1} \der{V_i}{q_1} \]
because of the mismatch between the time direction $t_1$ in which the variational derivative acts and the index $i$ of the Lagrangian coefficient. The multi-time Euler-Lagrange equations of the type 
\[ \var{i}{\cL_i}{q_i} = \var{j}{\cL_j}{q_j} \]
all reduce to the trivial equation $q_1 = q_1$, expressing the fact that all $\cL_i$ yield the same momentum.

Since $\cL_1$ is regular, $\det \left( \der{^2 \cL_1}{q_1^2} \right) \neq 0$, we can find a canonical Hamiltonian for the first equation by Legendre transformation,
\[ H_1(q,\pi) = \frac{1}{2} |\pi|^2 + V_1(q), \]
where we use $\pi$ to denote the cotangent space coordinate and $|\pi|^2 = \pi^T \pi$.

For $i \geq 2$ we consider $r = q_1$ as a second dependent variable. In other words, we double the dimension of the configuration space, which is now has coordinates $(q,r) = (q, q_1)$. The Lagrangians $\cL_i(q,r,q_i,r_i) = r q_i - V_i(q,r)$ are linear in velocities. We have $p(q,r) = r$, hence the symplectic form \eqref{symp} is
\[ \omega = \d r \wedge \d q . \]
This is the canonical symplectic form, with the momentum replaced by $r = q_1$. Hence we can consider $r$ as momentum, thus identifying the extended configuration space spanned by $q$ and $r$ with the phase space $T^*Q$. 

Applying Proposition \ref{prop-linlag}, we arrive at the following result:

\begin{thm}\label{thm-1form}
	The multi-time Euler-Lagrange equations of a 1-form with coefficients \eqref{L-1-1}--\eqref{L-1-2} are equivalent, under the identification $\pi = q_1$, to a system of Hamiltonian equations with respect to the canonical symplectic form $\omega = \d \pi \wedge \d q$, with Hamilton functions
	\[ H_1(q,\pi) = \frac{1}{2} |\pi|^2 + V_1(q) \qquad \text{and} \qquad H_i(q,\pi) = V_i(q,\pi) \quad \text{for } i \geq 2 \]
\end{thm}

\begin{example}
	From the Lagrangian 1-form for the Toda lattice given in Example \ref{ex-toda} we find
	\begin{align*}
		H_1 &= \sum_k \left( \frac{1}{2} \left( \pi^{[k]} \right)^2 + \exp\!\left( \q{k} - \q{k-1} \right) \right) , \\
		H_2 &= \sum_k \left( \frac{1}{3}\left( \pi^{[k]} \right)^3 + \left( \pi^{[k]}  + \pi^{[k-1]}  \right) \exp\!\left( \q{k} - \q{k-1} \right) \right) , \\
		H_3 &= \sum_k \bigg( \frac{1}{4} \left( \pi^{[k]} \right)^{4} + \left( \left( \pi^{[k+1]}  \right)^2 + \pi^{[k+1]}  \pi^{[k]} + \left( \pi^{[k]} \right)^2 \right) \exp\!\left( \q{k+1} - \q{k} \right) \\
		&\hspace{1.5cm} + \exp\!\left( \q{k+2} - \q{k} \right) + \frac{1}{2} \exp\!\left( 2 (\q{k+1} - \q{k}) \right) \bigg) ,\\
		&\vdotswithin{=}
	\end{align*}
\end{example}

We have limited the discussion in this section to the case where $\cL_1$ is quadratic in the velocity. There are some interesting examples that do not fall into this category, like the Volterra lattice, which has a Lagrangian linear in velocities, and the relativistic Toda lattice, which has a Lagrangian with a more complicated dependence on velocities (see e.g.\@ \cite{suris2003problem} and the references therein). The discussion above can be adapted to other types of Lagrangian 1-forms if one of its coefficients $\cL_i$ has an invertible Legendre transform, or if they are collectively Legendre-transformable as described in \cite{suris2013variational}. 

\subsection{From Hamiltonian to Pluri-Lagrangian systems}

The procedure from Section \ref{sec-to-ham1} can be reversed to construct a Lagrangian 1-form from a number of Hamiltonians.

\begin{thm}\label{thm-ham-to-lag}
	Consider Hamilton functions $H_i: T^*Q \rightarrow \R$, with $H_1(q,\pi) = \frac{1}{2} |\pi|^2 + V_1(q)$. Then the multi-time Euler-Lagrange equations of the Lagrangian 1-form $\sum_i \cL_i \, \d t_i$ with
	\begin{align*}
		\cL_1 &= \frac{1}{2} |q_1|^2 - V_1(q) \\
		\cL_i &= q_1 q_i - H_i(q,q_1) \qquad \text{for } i \geq 2
	\end{align*}
	are equivalent to the Hamiltonian equations under the identification $\pi = q_1$.
\end{thm}
\begin{proof}
	Identifying $\pi = q_1$, the multi-time Euler-Lagrange equations of the type \eqref{el1} are	
	\begin{align*}
		\var{1}{\cL_1}{q} &= 0 \quad\Leftrightarrow\quad q_{11} = - \der{V_1(q)}{q} , \\
		\var{i}{\cL_i}{q_1} &= 0 \quad\Leftrightarrow\quad q_i = \der{H_i(q,\pi)}{p}, \\
		\var{i}{\cL_i}{q} &= 0 \quad\Leftrightarrow\quad \pi_i = -\der{H_i(q,\pi)}{q}.
	\end{align*}
	The multi-time Euler-Lagrange equations of the type \eqref{el2} are trivially satisfied because
	\[ \var{i}{\cL_i}{q_i} = q_1 \]
	for all $i$.
\end{proof}
Note that the statement of Theorem \ref{thm-ham-to-lag} does not require the Hamiltonian equations to commute, i.e.\@ it is not imposed that the Hamiltonian vector fields $X_{H_i}$ associated to the Hamilton functions $H_i$ satisfy $[X_{H_i},X_{H_j}] = 0$. However, if they do not commute then for a generic initial condition $(q_0,\pi_0)$ there will be no solution $(q,\pi): \R^N \rightarrow T^*Q$ to the equations
\begin{align*}
	& \der{}{t_i} (q(t_1,\ldots t_N),\pi(t_1,\ldots t_N)) = X_{H_i} (q(t_1,\ldots t_N),\pi(t_1,\ldots t_N)) \qquad (i = 1, \ldots, N),\\
	&(q(0,\ldots,0),\pi(0,\ldots,0)) = (q_0,\pi_0) .
\end{align*}
Hence the relevance of Theorem \ref{thm-ham-to-lag} lies almost entirely in the case of commuting Hamiltonian equations. If they do not commute then it is an (almost) empty statement because neither the system of Hamiltonian equations nor the multi-time Euler-Lagrange equations will have solutions for generic initial data.

\begin{example}
	The Kepler Problem, describing the motion of a point mass around a gravitational center, is one of the classic examples of a completely integrable system. It possesses Poisson-commuting Hamiltonians $H_1,H_2,H_3: T^*\R^3 \rightarrow \R$ given by
	\begin{align*}
		H_1(q,\pi) &= \frac{1}{2} |\pi|^2 - |q|^{-1}, \quad&& \text{the energy, Hamiltonian for the physical motion,} \\
		H_2(q,\pi) &= (q \times \pi) \cdot \mathsf{e}_z, && \text{the 3rd component of the angular momentum, and} \\
		H_3(q,\pi) &= |q \times \pi|^2 , && \text{the squared magnitude of the angular momentum,}
	\end{align*}
	where $q = (x,y,z)$ and $\mathsf{e}_z$ is the unit vector in the $z$-direction. The corresponding coefficients of the Lagrangian 1-form  are
	\begin{align*}
		\cL_1 &= \frac{1}{2} |q_1|^2 + |q|^{-1} , \\
		\cL_2 &= q_1 \cdot q_2 - (q \times q_1)\cdot \mathsf{e}_z , \\
		\cL_3 &= q_1 \cdot q_3 - |q \times q_1|^2 .
	\end{align*}
	The multi-time Euler-Lagrange equations are	
	\[ \var{1}{\cL_1}{q} = 0 \quad\Rightarrow\quad q_{11} = \frac{q}{|q|^3},\]
	the physical equations of motion,
	\begin{align*}
		\var{2}{\cL_2}{q_1} = 0 &\quad\Rightarrow\quad q_2 = \mathsf{e}_z \times q, \\
		\var{2}{\cL_2}{q} = 0 &\quad\Rightarrow\quad q_{12} = - q_1 \times \mathsf{e}_z,\
	\end{align*}
	describing a rotation around the $z$-axis, and
	\begin{align*} 
		\var{3}{\cL_3}{q_1} = 0 &\quad\Rightarrow\quad q_3 = 2 (q \times q_1) \times q, \\
		\var{3}{\cL_3}{q} = 0 &\quad\Rightarrow\quad q_{13} = 2 (q \times q_1) \times q_1,
	\end{align*}
	describing a rotation around the angular momentum vector.
\end{example}

\subsection{Closedness and involutivity}
\label{sec-closed1}

In the pluri-Lagrangian theory, the exterior derivative $\d \cL$ is constant on solutions (see Proposition \ref{prop-almost-closed} in the Appendix). In many cases this constant is zero, i.e.\@ the Lagrangian 1-form is closed on solutions. Here we relate this property to the vanishing of Poisson brackets between the Hamilton functions.

\begin{prop}[{\cite[Theorem 3]{suris2013variational}}]
	Consider a Lagrangian 1-form $\cL$ as in Section \ref{sec-to-ham1} and the corresponding Hamilton functions $H_i$. On solutions to the multi-time Euler-Lagrange equations, and identifying $\pi = p(q,q_1) = \der{\cL_i}{q_i}$, there holds
	\begin{equation}\label{closedness1}
		\begin{split}
			\frac{\d \cL_j}{\d t_i} - \frac{\d \cL_i}{\d t_j} &= p_j q_i - p_i q_j  \\
			&= \{ H_j, H_i \} ,
		\end{split}
	\end{equation}
	where $ \{ \cdot , \cdot \}$ denotes the canonical Poisson bracket and $p_j$ and $q_j$ are shorthand for $\frac{\d p}{\d t_j}$ and $\frac{\d q}{\d t_j}$.
\end{prop}
\begin{proof}
	On solutions of the multi-time Euler-Lagrange equations there holds
	\begin{align*}
		\frac{\d \cL_j}{\d t_i} &= \der{\cL_j}{q} q_i + \der{\cL_j}{q_1} q_{1i} + \der{\cL_j}{q_j} q_{ij} \\
		&= \left( \frac{\d}{\d t_j} \der{\cL_j}{q} \right) q_i + \der{\cL_j}{q_j} q_{ij} \\
		&= p_j q_i + p q_{ij}.
	\end{align*}
	Hence
	\begin{equation}\label{dL1}
		\frac{\d \cL_j}{\d t_i} - \frac{\d \cL_i}{\d t_j} = p_j q_i - p_i q_j .
	\end{equation}
	Alternatively, we can calculate this expression using the Hamiltonian formalism. We have
	\begin{align*}
		\frac{\d \cL_j}{\d t_i} - \frac{\d \cL_i}{\d t_j}
		&= \frac{\d}{\d t_i} (p q_j - H_j) - \frac{\d}{\d t_j} ( p q_i - H_i) \\
		&= p_i q_j - p_j q_i + 2 \{ H_j, H_i \} .
	\end{align*}
	Combined with Equation \eqref{dL1}, this implies Equation \eqref{closedness1}.
\end{proof}

As a corollary we have:
\begin{thm}\label{thm-dL0-1}
	The Hamiltonians $H_i$ from Theorem \ref{thm-1form} are in involution if and only if $\d \cL = 0$ on solutions.
\end{thm}

All examples of Lagrangian 1-forms discussed so far satisfy $\d \cL = 0$ on solutions. This need not be the case.

\begin{example}
	Let us consider a system of commuting equations that is not Liouville integrable. Fix a constant $c \neq 0$ and consider the 1-form $\cL = \cL_1 \,\d t_1 + \cL_2 \,\d t_2$ with
	\[ \cL_1 \llbracket r,\theta \rrbracket = \frac{1}{2} r^2 \theta_1^2 + \frac{1}{2} r_1^2 - V(r) - c \theta, \]
	which for $c = 0$ would describe a central force in the plane governed by the potential $V$, and
	\[ \cL_2 \llbracket r,\theta \rrbracket = r^2 \theta_1 (\theta_2 - 1) + r_1 r_2 . \]
	Its multi-time Euler-Lagrange equations are 
	\begin{align*}
		&r_{11} = -V'(r) + r \theta_1^2 , \\
		&\frac{\d}{\d t_1} (r^2 \theta_1) = -c, \\
		&r_2 = 0, \\
		&\theta_2 = 1,
	\end{align*}
	and consequences thereof. Notably, we have
	\[ \frac{\d \cL_2}{\d t_1} - \frac{\d \cL_1}{\d t_2} = c \]
	on solutions, hence $\d \cL$ is nonzero.
	
	By Theorem \ref{thm-1form} the multi-time Euler-Lagrange equations are equivalent to the canonical Hamiltonian systems with
	\begin{align*}
		H_1(r,\theta,\pi,\sigma) &= \frac{1}{2} \frac{\sigma^2}{r^2} + \frac{1}{2} \pi^2 + V(r) + c \theta \\
		H_2(r,\theta,\pi,\sigma) & = \sigma ,
	\end{align*}
	where $\pi$ and $\sigma$ are the conjugate momenta to $r$ and $\theta$. The Hamiltonians are not in involution, but rather
	\[ \{H_2,H_1\} = c = \frac{\d \cL_2}{\d t_1} - \frac{\d \cL_1}{\d t_2} . \]
\end{example}

\section{Hamiltonian structure of Lagrangian 2-form systems}
\label{sec-2form}

In order to generalize the results from Section \ref{sec-1form} to the case of 2-forms, we need to carefully examine the relevant geometric structure. A useful tool for this is the variational bicomplex, which is also used in Appendix \ref{sec-appendix} to study the multi-time Euler-Lagrange equations.

\subsection{The variational bicomplex}
\label{sec-varbi} 

To facilitate the variational calculus in the pluri-Lagrangian setting, it is useful to consider the variation operator $\delta$ as an exterior derivative, acting in the fiber $J^\infty$ of the infinite jet bundle. We call $\delta$ the \emph{vertical exterior derivative} and $\d$, which acts in the base manifold $M$, the \emph{horizontal exterior derivative}. Together they provide a double grading of the space $\Omega(M \times J^\infty)$ of differential forms on the jet bundle. The space of \emph{$(a,b)$-forms} is generated by those $(a+b)$-forms structured as
\[ f\llbracket u \rrbracket \, \delta u_{I_1} \wedge \ldots \wedge \delta u_{I_a} \wedge \d t_{j_1} \ldots \wedge \d t_{j_b} . \]
We denote the space of $(a,b)$-forms by $\Omega^{(a,b)} \subset \Omega^{a+b}(M \times J^\infty)$. We call elements of $\Omega^{(0,b)}$  horizontal forms and elements of $\Omega^{(a,0)}$ vertical forms. The Lagrangian is a horizontal $d$-form, $\cL \in \Omega^{(0,d)}$.

The horizontal and vertical exterior derivatives are characterized by the anti-derivation property,
\begin{align*}
	\d \left( \omega_1^{p_1,q_1} \wedge \omega_2^{p_2,q_2} \right) 
	&= \d \omega_1^{p_1,q_1} \wedge \omega_2^{p_2,q_2} + (-1)^{p_1+q_1} \, \omega_1^{p_1,q_1} \wedge \d \omega_2^{p_2,q_2} , 
	\\
	\delta \left( \omega_1^{p_1,q_1} \wedge \omega_2^{p_2,q_2} \right) 
	&= \delta \omega_1^{p_1,q_1} \wedge \omega_2^{p_2,q_2} + (-1)^{p_1+q_1} \, \omega_1^{p_1,q_1} \wedge \delta \omega_2^{p_2,q_2} ,
\end{align*}
where the upper indices denote the type of the forms, and by the way they act on $(0,0)$-forms, and basic $(1,0)$ and $(0,1)$-forms:
\begin{alignat*}{3}
	\d f\llbracket u \rrbracket &= \sum_j \partial_j f\llbracket u \rrbracket \,\d t_j ,
	& \delta f\llbracket u \rrbracket &= \sum_I \der{f\llbracket u \rrbracket}{u_I} \delta u_{I} ,
	\\
	\d (\delta u_I ) &= - \sum_j \delta u_{Ij} \wedge \d t_j ,
	\hspace{2cm} & \delta(\delta u_I) &= 0 ,
	\\
	\d(\d t_j) &= 0 ,
	& \delta (\d t_j) &= 0 .
\end{alignat*}
One can verify that $\d + \delta: \Omega^{a+b} \rightarrow \Omega^{a+b+1}$ is the usual exterior derivative and that
\[ \delta^2 = \d^2 = \delta \d + \d \delta = 0 . \]
Time-derivatives $\partial_j$ act on vertical forms as $\partial_j(\delta u_I) = \delta u_{Ij}$, on horizontal forms as $\partial_j(\d t_k) = 0$, and obey the Leibniz rule with respect to the wedge product. As a simple but important example, note that
\[ \d (f\llbracket u \rrbracket \, \delta u_I)
= \sum_{j=1}^N \partial_j f\llbracket u \rrbracket \, \d t_j \wedge \delta u_I - f\llbracket u \rrbracket \, \delta u_{I t_j} \wedge \d t_j
= \sum_{j=1}^N - \partial_j (f\llbracket u \rrbracket \, \delta u_I) \wedge \d t_j .
\]

The spaces $\Omega^{(a,b)}$, for $a \geq 0$ and $0 \leq b \leq N$, related to each other by the maps $\d$ and $\delta$, are collectively known as the \emph{variational bicomplex} \cite[Chapter 19]{dickey2003soliton}. A slightly different version  of the variational bicomplex, using contact 1-forms instead of vertical forms, is presented in \cite{anderson1992introduction}. We will not discuss the rich algebraic structure of the variational bicomplex here.

For a horizontal $(0,d)$-form $\cL\llbracket u \rrbracket$, the variational principle
\[  \delta \int_\Gamma \cL\llbracket u \rrbracket = \delta \int_\Gamma \sum_{i_1 < \ldots < i_d} \cL_{i_1,\ldots,i_d}\llbracket u \rrbracket \, \d t_{i_1} \wedge \ldots \wedge \d t_{i_d} =  0 \]
can be understood as follows. Every vertical vector field $V = v(t_1,\ldots,t_a) \frac{\partial}{\partial u}$, such that its \emph{prolongation}
\[ \pr V = \sum_I v_I \frac{\partial}{\partial u_I} \]
vanishes on the boundary $\partial \Gamma$, must satisfy
\[ \int_\Gamma \iota_{\pr V} \delta \cL = \int_\Gamma \sum_{i_1 < \ldots < i_d} \iota_{\pr V} ( \delta \cL_{i_1,\ldots,i_d}\llbracket u \rrbracket ) \, \d t_{i_1} \wedge \ldots \wedge \d t_{i_d} = 0 . \]
Note that the integrand is a horizontal form, so the integration takes place on $\Gamma \subset M$, independent of the bundle structure.

\subsection{The space of functionals and its pre-symplectic structure}

In the rest of our discussion, we will single out the variable $t_1 = x$ and view it as the space variable, as opposed to the time variables $t_2,\ldots,t_N$. For ease of presentation we will limit the discussion here to real scalar fields, but it is easily extended to complex or vector-valued fields. We consider functions $u: \R \rightarrow \R: x \mapsto u(x)$ as fields at a fixed time.  Let $J^\infty$ be the fiber of the corresponding infinite jet bundle, where the prolongation of $u$ has coordinates $[u] = (u,u_x,u_{xx},\ldots)$. Consider the space of functions of the infinite jet of $u$,
\[ \cV =  \left\{ v: J^\infty \rightarrow \R \right\} . \]
Note that the domain $J^\infty$ is the fiber of the jet bundle, hence the elements $v \in \cV$ depend on $x$ only through $u$. We will be dealing with integrals $\int v \,\d x$ of elements $v \in \cV$. In order to avoid convergence questions, we understand the symbol $\int v \,\d x$ as a \emph{formal integral}, defined as the equivalence class of $v$ modulo space-derivatives. In other words, we consider the space of functionals
\[ \cF = \cV \big/ \partial_x \! \cV , \]
where
\[ \partial_x = \frac{\d}{\d x} = \sum_I  u_{Ix} \der{}{u_I} . \]

The variation of an element of $\cF$ is computed as
\begin{equation}\label{ext-der}
	\delta \int v \,\d x = \int \var{}{v}{u} \,\delta u \wedge \d x ,
\end{equation}
where
\[ \var{}{}{u} = \sum_{\alpha = 0}^\infty (-1)^{\alpha} \partial_x^{\alpha} \der{}{u_{x^{\alpha} }} . \]
Equation \eqref{ext-der} is independent of the choice of representative $v \in \cV$ because the variational derivative of a full $x$-derivative is zero.

Since $\cV$ is a linear space, its tangent spaces can be identified with $\cV$ itself. In turn, every $v \in \cV$ can be identified with a vector field $v \frac{\partial}{\partial u}$. We will define Hamiltonian vector fields in terms of $\cF$-valued forms on $\cV$. An $\cF$-valued 1-form $\theta$ can be represented as the integral of a $(1,1)$-form in the variational bicomplex,
\[ \theta = \int\sum_k a_k[u] \, \delta u_{x^k} \wedge \d x \]
and defines a map
\[ \cV \rightarrow \cF: v \mapsto \iota_{v} \theta = \int \sum_k a_k[u] \,\partial_x^k v[u] \,\d x . \]
This amounts to pairing the 1-form with the infinite jet prolongation of the vector field $v \frac{\partial}{\partial u}$. Note that $\cF$-valued forms are defined modulo $x$-derivatives: $\int \partial_x \theta \wedge \d x = 0$ because its pairing with any vector field in $\cV$ will yield a full $x$-derivative, which represents the zero functional in $\cF$. Hence the space of $\cF$-valued 1-forms is $\Omega^{(1,1)} / \partial_x \Omega^{(1,1)}$.

An $\cF$-valued 2-form
\[ \omega = \int \sum_{k,\ell} a_{k,\ell}[u] \, \delta u_{x^k} \wedge \delta u_{x^\ell} \wedge \d x \]
defines a skew-symmetric map
\[ \cV \times \cV \rightarrow \cF: (v,w) \mapsto \iota_w \iota_v \omega = \int \sum_{k,\ell} a_{k,\ell}[u] \left(\partial_x^k v[u] \, \partial_x^\ell w[u] - \partial_x^k w[u] \, \partial_x^\ell v[u] \right) \d x \]
as well as a map from vector fields to $\cF$-valued 1-forms
\[ \cV \rightarrow \Omega^{(1,1)} / \partial_x \Omega^{(1,1)}: v \mapsto \iota_v \omega = \int \sum_{k,\ell} a_{k,\ell}[u] \left(\partial_x^k v[u] \, \delta u_{x^\ell} - \partial_x^\ell v[u] \, \delta u_{x^k} \right) \wedge \d x . \]

\begin{definition}
	A closed $(2,1)$-form $\omega$ on $\cV$ is called \emph{pre-symplectic}.
\end{definition}

Equivalently we can require the form to be vertically closed, i.e.\@ closed with respect to $\delta$. Since the horizontal space is 1-dimensional ($x$ is the only independent variable) every $(a,1)$-form is closed with respect to the horizontal exterior derivative $\d$, so only vertical closedness is a nontrivial property.

We choose to work with pre-symplectic forms instead of symplectic forms, because the non-degeneracy required of a symplectic form is a subtle issue in the present context. Consider for example the pre-symplectic form $\omega = \int \delta u \wedge \delta u_x \wedge \d x$. It is degenerate because
\[ \int \iota_{v} \omega = \int (v \, \delta u_x - v_x \, \delta u) \wedge \d x = \int - 2 v_x \, \delta u \wedge \d x , \]
which is zero whenever $v[u]$ is constant. However, if we restrict our attention to compactly supported fields, then a constant must be zero, so the restriction of $\omega$ to the space of compactly supported fields is non-degenerate.

\begin{definition}
	A \emph{Hamiltonian vector field} with Hamilton functional $\tint H \,\d x$ is an element $v \in \cV$ satisfying the relation
	\[ \int \iota_{v} \omega = \int \delta H \wedge \d x . \]
\end{definition}

Note that if $\omega$ is degenerate, we cannot guarantee existence or uniqueness of a Hamiltonian vector field in general.

\subsection{From pluri-Lagrangian to Hamiltonian systems}
\label{sec-pL2Ham}

We will consider two different types of Lagrangian 2-forms. The first type are those where for every $j$ the coefficient $\cL_{1j}$ is linear in $u_{t_j}$. This is the case for the 2-form for the potential KdV hierarchy from Example \ref{ex-kdv} and for the Lagrangian 2-forms of many other hierarchies like the AKNS hierarchy \cite{sleigh2019lax} and the modified KdV, Schwarzian KdV and Krichever-Novikov hierarchies \cite{vermeeren2019variational}. The second type satisfy the same property for $j > 2$, but have a coefficient $\cL_{12}$ that is quadratic in $u_{t_2}$, as is the case for the  Boussinesq hierarchy from Example \ref{ex-bsq}.

\subsubsection[Linear]{When all $\cL_{1j}$ are linear in $u_{t_j}$}

Consider a Lagrangian 2-form $\cL\llbracket u \rrbracket = \sum_{i < j} \cL_{ij}\llbracket u \rrbracket \,\d t_i \wedge \d t_j$, where for all $j$ the variational derivative $\var{1}{\cL_{1j}}{u_{t_j}}$ does not depend on any $t_j$-derivatives, hence we can write 
\[ \var{1}{\cL_{1j}}{u_{t_j}} = p[u] \]
for some function $p[u]$ depending on on an arbitrary number of space derivatives, but not on any time-derivatives. We use single square brackets $[\cdot]$ to indicate dependence on space derivatives only. Note that $p$ does not depend on the index $j$. This is imposed on us by the multi-time Euler-Lagrange equation stating that $\var{1}{\cL_{1j}}{u_{t_j}}$ is independent of $j$. 

Starting from these assumptions and possibly adding a full $x$-derivative (recall that $x = t_1$) we find that the coefficients $\cL_{1j}$ are of the form
\begin{equation}\label{L-2-1}
	\cL_{1j}\llbracket u \rrbracket = p[u] u_j - h_j[u] ,
\end{equation}
where $u_j$ is shorthand notation for the derivative $u_{t_j}$. Coefficients of this form appear in many prominent examples, like the potential KdV hierarchy and several hierarchies related to it \cite{suris2016lagrangian, vermeeren2019continuum, vermeeren2019variational} as well as the AKNS hierarchy \cite{sleigh2019lax}.
Their Euler-Lagrange equations are
\begin{equation}\label{EL2}
	\cE_p u_j - \var{1}{h_j[u]}{u} = 0 ,
\end{equation}
where $\cE_p$ is the differential operator
\[ \cE_p = \sum_{k = 0}^\infty \left( (-1)^k \partial_x^k \der{p}{u_{x^k}} - \der{p}{u_{x^k}} \partial_x^k \right) . \]
We can also write $\cE_p = \mathsf{D}_p^* - \mathsf{D}_p$, where $\mathsf{D}_p$ is the Fréchet derivative of $p$ and $\mathsf{D}_p^*$ its adjoint \cite[Eqs (5.32) resp.\@ (5.79)]{olver2000applications}.

Consider the pre-symplectic form
\begin{equation}\label{ms}
	\begin{split}
		\omega &= -\delta p[u] \wedge \delta u \wedge \d x \\
		&= - \sum_{k = 1}^\infty \der{p}{u_{x^k}} \delta u_{x^k} \wedge \delta u \wedge \d x .
	\end{split}
\end{equation}
Inserting the vector field $X = \chi \der{}{u}$ we find
\begin{align*}
	\int \iota_X \omega
	&= \int \sum_{k = 0}^\infty \left( \der{p}{u_{x^k}}  \chi \, \delta u_{x^k} \wedge \d x - \der{p}{u_{x^k}} \chi_{x^k} \, \delta u \wedge \d x \right) \\
	&= \int \sum_{k = 0}^\infty \left( (-1)^k \partial_x^k \!\left( \der{p}{u_{x^k}} \chi \right) - \der{p}{u_{x^k}} \chi_{x^k} \right) \delta u \wedge \d x  \\
	&= \int \cE_p \chi \, \delta u \wedge \d x .
\end{align*}
From the Hamiltonian equation of motion
\[ \int \iota_X \omega = \int \delta h_j[u] \wedge \d x \]
we now obtain that the Hamiltonian vector field $X = \chi \der{}{u}$ associated to $h_j$ satisfies
\[ \cE_p \chi = \var{1}{h_j}{u} ,\]
which corresponds the Euler-Lagrange equation \eqref{EL2} by identifying $\chi = u_{t_j}$. This observation was made previously in the context of loop spaces in \cite[Section 1.3]{mokhov1998symplectic}.

The Poisson bracket associated to the symplectic operator $\cE_p$ is formally given by
\begin{equation}\label{poisson}
	\left\{ \tint f \,\d x, \tint g \,\d x \right\} = - \int \var{}{f}{u} \, \cE_p^{-1} \, \var{}{g}{u} \,\d x. 
\end{equation}
If the pre-symplectic form is degenerate, then $\cE_p$ will not be invertible. In this case $\cE_p^{-1}$ can be considered as a pseudo-differential operator and the Poisson bracket is called \emph{non-local} \cite{mokhov1998symplectic,desole2013nonlocal}. Note that $\{ \cdot , \cdot \}$ does not satisfy the Leibniz rule because there is no multiplication on the space $\cF$ of formal integrals. However, we can recover the Leibniz rule in one entry by introducing 
\[ [f , g ] = - \sum_{k=0}^\infty \der{f}{u_{x^k}} \partial_x^k \, \cE_p^{-1} \, \var{}{g}{u}. \]
Then we have
\[ \left\{ \tint f \,\d x , \tint g \,\d x \right\} = \int [f,g] \,\d x \]
and 
\[ [f g, h] = f [g, h] + [f,h] g .\]

In summary, we have the following result:
\begin{thm}\label{thm-ham}
	Assume that $\var{1}{h_j[u]}{u}$ is in the image of $\cE_p$ and has a unique inverse (possibly in some equivalence class) for each $j$. Then the evolutionary PDEs 
	\[ u_j = \cE_p^{-1} \var{1}{h_j[u]}{u},\]
	which imply the Euler-Lagrange equations \eqref{EL2} of the Lagrangians \eqref{L-2-1}, are Hamiltonian with respect to the symplectic form \eqref{ms} and the Poisson bracket \eqref{poisson}, with Hamilton functions $h_j$.
\end{thm}

This theorem applies without assuming any kind of consistency of the system of multi-time Euler-Lagrange equations. Of course we are mostly interested in the case where the multi-time Euler-Lagrange equations are equivalent to an integrable hierarchy. In almost all known examples (see e.g.\@ \cite{suris2016lagrangian,sleigh2019lax,vermeeren2019variational}) the multi-time Euler-Lagrange equations consist of an integrable hierarchy in its evolutionary form and differential consequences thereof. Hence the general picture suggested by these examples is that the multi-time Euler-Lagrange equations are equivalent to the equations  $u_j = \cE_p^{-1} \var{1}{h_j[u]}{u}$ form Theorem \ref{thm-ham}. In light of these observations, we emphasize the following consequence of Theorem \ref{thm-ham}.

\begin{cor}
	If the multi-time Euler-Lagrange equations are evolutionary, then they are Hamiltonian.
\end{cor}

\begin{example}
	The pluri-Lagrangian structure for the potential KdV hierarchy, given in Example \ref{ex-kdv}, has $p = \frac{1}{2} u_x$. Hence
	\[ \cE_p = -\partial_x \der{p}{u_x} - \der{p}{u_x} \partial_x = - \partial_x \]
	and
	\[ \left\{ \tint f \,\d x , \tint g \,\d x \right\} = \int \var{}{f}{u} \partial_x^{-1} \var{}{g}{u} \,\d x . \]
	Here we assume that $\var{}{g}{u}$ is in the image of $\partial_x$. Then $ \partial_x^{-1} \var{}{g}{u}$ is uniquely defined by the convention that it does not contain a constant term.
	If $f$ and $g$ depend only on derivatives of $u$, not on $u$ itself, this becomes the Gardner bracket \cite{gardner1971korteweg}
	\[ \left\{ \tint f \,\d x , \tint g \,\d x \right\} = \int \left(\partial_x \var{}{f}{u_x} \right) \var{}{g}{u_x} \,\d x. \]
	The Hamilton functions are
	\begin{align*}
		h_2[u] &= \frac{1}{2} u_{x} u_{t_2} - \cL_{12} =  u_{x}^3 + \frac{1}{2} u_{x} u_{xxx} , \\
		h_3[u] &= \frac{1}{2} u_{x} u_{t_3} - \cL_{13} =  \frac{5}{2} u_{x}^4 - 5 u_{x} u_{xx}^2 + \frac{1}{2} u_{xxx}^2, \\
		&\vdotswithin{=}
	\end{align*}
	
	A related derivation of the Gardner bracket from the multi-symplectic perspective was given in \cite{gotay1988multisymplectic}. It can also be obtained from the Lagrangian structure by Dirac reduction \cite{macfarlane1982equations}.
\end{example}

\begin{example}
	The Schwarzian KdV hierarchy, 
	\begin{align*}
		u_{2} &= -\frac{3 u_{11}^2}{2 u_{1}} + u_{111} , \\
		u_{3} &= -\frac{45 u_{11}^4}{8 u_{1}^3} + \frac{25 u_{11}^2 u_{111}}{2 u_{1}^2} - \frac{5 u_{111}^2}{2 u_{1}} - \frac{5 u_{11} u_{1111}}{u_{1}} + u_{11111} , \\
		&\vdotswithin{=}
	\end{align*}
	has a pluri-Lagrangian structure with coefficients \cite{vermeeren2019continuum}
	\begin{align*}
		\cL_{12} &= \frac{u_{3}}{2 u_{1}} - \frac{u_{11}^2}{2 u_{1}^2} , \\
		\cL_{13} &= \frac{u_{5}}{2 u_{1}} - \frac{3 u_{11}^4}{8 u_{1}^4} + \frac{u_{111}^2}{2 u_{1}^2} , \\
		\cL_{23} &= -\frac{45 u_{11}^6}{32 u_{1}^6} + \frac{57 u_{11}^4 u_{111}}{16 u_{1}^5} - \frac{19 u_{11}^2 u_{111}^2}{8 u_{1}^4} + \frac{7 u_{111}^3}{4 u_{1}^3} - \frac{3 u_{11}^3 u_{1111}}{4 u_{1}^4} - \frac{3 u_{11} u_{111} u_{1111}}{2 u_{1}^3} \\
		&\quad + \frac{u_{1111}^2}{2 u_{1}^2} + \frac{3 u_{11}^2 u_{11111}}{4 u_{1}^3} - \frac{u_{111} u_{11111}}{2 u_{1}^2} + \frac{u_{111} u_{113}}{u_{1}^2} - \frac{3 u_{11}^3 u_{13}}{2 u_{1}^4} + \frac{2 u_{11} u_{111} u_{13}}{u_{1}^3}\\
		&\quad - \frac{u_{1111} u_{13}}{u_{1}^2} + \frac{u_{11} u_{15}}{u_{1}^2} - \frac{27 u_{11}^4 u_{3}}{16 u_{1}^5} + \frac{17 u_{11}^2 u_{111} u_{3}}{4 u_{1}^4} - \frac{7 u_{111}^2 u_{3}}{4 u_{1}^3} - \frac{3 u_{11} u_{1111} u_{3}}{2 u_{1}^3} \\
		&\quad  + \frac{u_{11111} u_{3}}{2 u_{1}^2} + \frac{u_{11}^2 u_{5}}{4 u_{1}^3} - \frac{u_{111} u_{5}}{2 u_{1}^2}, \\
		&\vdotswithin{=}
	\end{align*}
	In this example we have $p = \frac{1}{2 u_x}$, hence
	\[ \cE_p = -\partial_x \der{p}{u_x} - \der{p}{u_x} \partial_x 
	= \frac{1}{u_x^2} \partial_x - \frac{u_{xx}}{u_x^3}
	= \frac{1}{u_x} \partial_x \frac{1}{u_x} \]
	and 
	\[ \cE_p^{-1} = u_x \partial_x^{-1} u_x . \]
	This nonlocal operator seems to be the simplest Hamiltonian operator for the SKdV equation, see for example \cite{dorfman1987dirac, wilson1988quasi}. The Hamilton functions for the first two equations of the hierarchy are
	\[ 
	h_2 = \frac{u_{11}^2}{2 u_{1}^2} \qquad \text{and} \qquad
	h_3 = \frac{3 u_{11}^4}{8 u_{1}^4} - \frac{u_{111}^2}{2 u_{1}^2} .
	\]
\end{example}

\subsubsection[Quadratic]{When $\cL_{12}$ is quadratic in $u_{t_2}$}

Consider a Lagrangian 2-form $\cL\llbracket u \rrbracket = \sum_{i < j} \cL_{ij}\llbracket u \rrbracket \,\d t_i \wedge \d t_j$ with
\begin{equation}\label{L-2-2a}
	\cL_{12} = \frac{1}{2}\alpha[u] u_2^2 - V[u],
\end{equation}
and, for all $j \geq 3$, $\cL_{1j}$ of the form
\begin{equation}\label{L-2-2b}
	\cL_{1j}\llbracket u \rrbracket = \alpha[u] u_2 u_j - h_j[u,u_2] ,
\end{equation}
where $[u,u_2] = (u,u_2,u_1,u_{12},u_{11},u_{112},\ldots)$ since the single bracket $[\cdot]$ denotes dependence on the fields and their $x$-derivatives only (recall that $x = t_1$). To write down the full set of multi-time Euler-Lagrange equations we need to specify all $\cL_{ij}$, but for the present discussion it is sufficient to consider the equations
\[ \var{12}{\cL_{12}}{u} = 0 \quad \Leftrightarrow \quad
\alpha[u] u_{22} = - \frac{\d \alpha[u]}{\d t_2} u_2 + \frac{1}{2} \sum_{k=0}^\infty (-1)^k \partial_x^k \left(\der{\alpha[u]}{u_{x^k}} u_2^2 \right) - \var{1}{V[u]}{u} \]
and
\[ \var{1j}{\cL_{1j}}{u_2} = 0 \quad \Leftrightarrow \quad
\alpha[u] u_{j} = \var{1}{h_j[u,u_2]}{u_2} . \]
We assume that all other multi-time Euler-Lagrange equations are consequences of these.

Since $\cL_{12}$ is non-degenerate, the Legendre transform is invertible and allows us to identify $\pi = \alpha[u] u_2$. Consider the canonical symplectic form on formal integrals, where now the momentum $\pi$ enters as a second field,
\[ \omega = \delta \pi \wedge \delta u \wedge \d x . \]
This defines the Poisson bracket
\begin{equation}\label{poisson2} 
	\left\{ \tint f \,\d x , \tint g \,\d x \right\} = - \int \left( \var{}{f}{\pi} \var{}{g}{u} - \var{}{f}{u} \var{}{g}{\pi} \right) \d x.
\end{equation}

The coefficients $\cL_{1j}\llbracket u \rrbracket = \alpha[u] u_2 u_j - h_j[u,u_2]$ are linear in their velocities $u_j$, hence they are Hamiltonian with respect to the pre-symplectic form 
\[ \delta (\alpha[u] u_2) \wedge \delta u \wedge \d x , \]
which equals $\omega$ if we identify $\pi = \alpha[u] u_2$. Hence we find the following result.

\begin{thm}\label{thm-ham2}
	A hierarchy described by a Lagrangian 2-form with coefficients of the form \eqref{L-2-2a}--\eqref{L-2-2b} is Hamiltonian with respect to the canonical Poisson bracket \eqref{poisson2}, with Hamilton functions 
	\[ H_2[u,\pi] = \frac{1}{2} \frac{\pi^2}{\alpha[u]} + V[u] \]
	and
	\[ H_j[u,\pi] = h_j \! \left[ u,\frac{\pi}{\alpha[u]} \right] \]
	for $j \geq 3$.	
\end{thm}

\begin{example}
	The Lagrangian 2-form for the Boussinesq hierarchy from Example \ref{ex-bsq} leads to
	\begin{align*}
		H_{2} &= \frac{1}{2} \pi^2 + 2 u_{1}^3 + \frac{3}{2} u_{11}^2 , \\
		H_{3} &= - 6 u_{1}^4 - 27 u_{1} u_{11}^2 + 6 u_{} \pi_{1} \pi - \frac{9}{2} u_{111}^2 - \frac{3}{2} \pi_1^2 ,
	\end{align*}
	where the Legendre transform identifies $\pi = u_2$. Indeed we have
	\begin{align*}
		\left\{ \tint H_{2} \,\d x , \tint u \,\d x \right\} &= \tint \pi \,\d x  = \tint u_2\,\d x  ,\\
		\left\{ \tint H_{2} \,\d x , \tint \pi \,\d x \right\} &= \tint (12 u_1 u_{11} - 3 u_{111}) \,\d x = \tint \pi_2 \,\d x ,
	\end{align*}
	and 
	\begin{align*}
		\left\{ \tint H_{3} \,\d x , \tint u \,\d x \right\} &= \tint (-6 u_1 \pi + 3 \pi_{11}) \,\d x = \tint u_3 \,\d x , \\
		\left\{ \tint H_{3} \,\d x , \tint \pi \,\d x \right\}& = \tint \left( -72 u_1^2u_{11} + 108 u_{11}u_{111} + 54 u_1 u_{1111} - 6 \pi \pi_1 - 9 u_{111111} \right) \d x \\
		&= \tint \pi_3 \,\d x .
	\end{align*}
\end{example}

\subsection{Closedness and involutivity}
\label{sec-closed2}

Let us now have a look at the relation between the closedness of the Lagrangian 2-form and the involutivity of the corresponding Hamiltonians.

\begin{prop}\label{prop-triple1}
	On solutions of the multi-time Euler-Lagrange equations of a Lagrangian 2-form with coefficients $\cL_{1j}$ given by Equation \eqref{L-2-1}, there holds
	\begin{equation}\label{triple1}
		\{ h_i, h_j \} = \int \left( p_i u_j - p_j u_i \right) \d x = \int \left( \frac{\d \cL_{1i}}{\d t_j} - \frac{\d \cL_{1j}}{\d t_i} \right) \d x,
	\end{equation}
	where the Poisson bracket is given by Equation \eqref{poisson}.
\end{prop}
\begin{proof}
	On solutions of the Euler-Lagrange equations we have
	\begin{align*}
		\int \frac{\d \cL_{1i}}{\d t_j} \,\d x
		&=\int \left( \var{1}{\cL_{1i}}{u} u_j  + \der{\cL_{1i}}{u_i} u_{ij} \right) \d x \\
		&=\int \left( \left( \frac{\d}{\d t_i}\var{1i}{\cL_{1i}}{u_i} \right) u_j  + \der{\cL_{1i}}{u_i} u_{ij} \right) \d x \\
		&=\int \left( p_i u_j + p u_{ij} \right) \d x .
	\end{align*}
	It follows that
	\begin{equation}\label{pv1}
		\int \left( \frac{\d \cL_{1i}}{\d t_j} - \frac{\d \cL_{1j}}{\d t_i} \right) \d x = \int \left( p_i u_j - p_j u_i \right) \d x .
	\end{equation}
	On the other hand we have that
	\begin{align*}
		\int \left( \frac{\d \cL_{1i}}{\d t_j} - \frac{\d \cL_{1j}}{\d t_i} \right) \d x
		&= \int \left( \frac{\d}{\d t_j} (p u_i - h_i) - \frac{\d}{\d t_i} (p u_j - h_j) \right) \d x \\
		&= - \int \left( p_i u_j - p_j u_i \right) \d x + 2 \{ h_i, h_j \} .
	\end{align*}
	Combined with Equation \eqref{pv1}, this implies both identities in Equation \eqref{triple1}.
\end{proof}

\begin{prop}\label{prop-triple2}
	On solutions of the multi-time Euler-Lagrange equations of a Lagrangian 2-form with coefficients $\cL_{1j}$ given by Equations \eqref{L-2-2a}--\eqref{L-2-2b}, there holds
	\[
	\{ H_i, H_j \} = \int \left( \pi_i u_j - \pi_j u_i \right) \d x = \int \left( \frac{\d \cL_{1i}}{\d t_j} - \frac{\d \cL_{1j}}{\d t_i} \right) \d x,
	\]
	where the Poisson bracket is given by Equation \eqref{poisson2} and the Hamilton functions $H_j$ are given in Theorem \ref{thm-ham2}.
\end{prop}
\begin{proof}
	Analogous to the proof of Proposition \ref{prop-triple1}, with $p[u]$ replaced by the field $\pi$.
\end{proof}

\begin{thm}\label{thm-dL0-2}
	Let $\cL$ be a Lagrangian 2-form with coefficients $\cL_{1j}$ given by Equation \eqref{L-2-1} or by Equations \eqref{L-2-2a}--\eqref{L-2-2b}. Consider the corresponding Hamiltonian structures, given by $H_{1j} = h_j$ or $H_{1j} = H_j$, as in Theorems \ref{thm-ham} and \ref{thm-ham2} respectively. There holds $\{ H_{1i}, H_{1j} \} = 0$ if and only if
	\[ \int \left( \frac{\d \cL_{ij}}{\d t_1} - \frac{\d \cL_{1j}}{\d t_i} + \frac{\d \cL_{1i}}{\d t_j} \right) \d x = 0 \]
	on solutions of the multi-time Euler-Lagrange equations.
\end{thm}
\begin{proof}
	Recall that $t_1 = x$, hence $\frac{\d}{\d t_1} = \partial_x$. By definition of the formal integral as an equivalence class, we have $\int \partial_x \cL_{ij} \,\d x = 0$. Hence the claim follows from Proposition \ref{prop-triple1} or Proposition \ref{prop-triple2}.
\end{proof}

It is known that $\d \cL \llbracket u \rrbracket$ is constant in the set of solutions $u$ to the multi-time Euler-Lagrange equations (see Proposition \ref{prop-almost-closed}). In most examples, one can verify using a trivial solution that this constant is zero.

\begin{cor}
	If a Lagrangian 2-form, with coefficients $\cL_{1j}\llbracket u \rrbracket$ given by Equation \eqref{L-2-1} or by Equations \eqref{L-2-2a}--\eqref{L-2-2b}, is closed for a solution $u$ to the pluri-Lagrangian problem, then $\{ H_{1i}, H_{1j} \} = 0$ for all $i,j$.
\end{cor}

All examples of Lagrangian 2-forms discussed so far satisfy $\d \cL = 0$ on solutions. We now present a system where this is not the case.

\begin{example}
	Consider a perturbation of the Boussinesq Lagrangian, obtained by adding $c u$ for some constant $c \in \R$,
	\[ \cL_{12} = \frac{1}{2} u_{2}^2 - 2 u_{1}^3 - \frac{3}{2} u_{11}^2 + c u , \]
	combined with the Lagrangian coefficients
	\begin{align*}
		\cL_{13} &= u_2 (u_3-1) \\
		\cL_{23} &= (6 u_1^2 - 3 u_{111}) (u_3-1) .
	\end{align*}
	The corresponding multi-time Euler-Lagrange equations consist of a perturbed Boussinesq equation,
	\[ u_{22} = 12 u_1 u_{11} - 3 u_{1111} + c \]
	and 
	\[ u_3 = 1. \]
	We have
	\[ \frac{\d \cL_{12}}{\d t_3} - \frac{\d \cL_{13}}{\d t_2} + \frac{\d \cL_{23}}{\d t_1} = c \]
	on solutions, hence $\d \cL$ is nonzero.
	
	The multi-time Euler-Lagrange equations are equivalent to the canonical Hamiltonian systems with
	\begin{align*}
		H_{2} &= \frac{1}{2} \pi^2 + 2 u_{1}^3 + \frac{3}{2} u_{11}^2 - cu \\
		H_{3} &= \pi .
	\end{align*}
	They are not in involution, but rather
	\[ \left\{\tint H_2 \,\d x , \tint H_3 \,\d x \right\} = \int \left( 12 u_{11}u_1 - 3 u_{1111} + c \right) \d x = \int c \, \d x . \]
	
	Note that if we would allow the fields in $\cV$ to depend explicitly on $x$, then we would find $\tint c \,\d x = \tint \partial_x(c x) \,\d x = 0$. Note that this is not a property of the Lagrangian form, but of the function space we work in. Allowing fields that depend on $x$ affects the definition of the formal integral $\tint (\cdot) \,\d x$ as an equivalence class modulo $x$-derivatives. If dependence on $x$ is allowed, then there is no such thing as a nonzero constant functional in this equivalence class. However, in our definition of $\cV$, fields can only depend on $x$ through $u$, hence $c$ is not an $x$-derivative and $\tint c \, \d x$ is not the zero element of $\cF$.
\end{example}

\subsection{Additional (nonlocal) Poisson brackets}
\label{sec-additional}

Even though the closedness property in Section \ref{sec-closed2} involves all coefficients of a Lagrangian 2-form $\cL$, so far we have only used the first row of coefficients $\cL_{1j}$ to construct Hamiltonian structures. A similar procedure can be carried out for other $\cL_{ij}$, but the results are not entirely satisfactory. In particular, it will not lead to true bi-Hamiltonian structures. Because of this slightly disappointing outcome, we will make no effort to present the most general statement possible. Instead we make some convenient assumptions on the form of the coefficients $\cL_{ij}$.

Consider a Lagrangian 2-form $\cL$ such that for all $i<j$ the coefficient $\cL_{ij}$ only contains derivatives with respect to $t_1$, $t_i$ and $t_j$ (no ``alien derivatives'' in the terminology of \cite{vermeeren2019continuum}). In addition, assume that $\cL_{ij}$ can be written as the sum of terms that each contain at most one derivative with respect to $t_i$ (if $i>1$) or $t_j$. In particular, $\cL_{ij}$ does not contain higher derivatives with respect to  $t_i$ (if $i>1$) or $t_j$, but mixed derivatives with respect to $t_1$ and $t_i$ or $t_1$ and $t_j$ are allowed. There is no restriction on the amount of $t_1$-derivatives. 

To get a Hamiltonian description of the evolution along the time direction $t_j$ from the Lagrangian $\cL_{ij}$, we should consider both $t_1$ and $t_i$ as space coordinates. Hence we will work on the space
\[ \cV \big/ \left( \partial_1 \! \cV + \partial_i \! \cV \right) . \]
For $i > 1$, consider the momenta
\[ p^{[i]}[u] = \var{1i}{\cL_{ij}}{u_j} . \]
From the assumption that each term of $\cL_{ij}$ contains at most one time-derivative it follows that $p^{[i]}$ only depends on $u$ and its $x$-derivatives. Note that  $p^{[i]}$ is independent of $j$ because of the multi-time Euler-Lagrange equation \eqref{2el2}. The variational derivative in the definition of $p^{[i]}$ is in the directions $1$ and $i$, corresponding to the formal integral, whereas the Lagrangian coefficient has indices $i$ and $j$. However, we can also write 
\[ p^{[i]}[u] = \var{1}{\cL_{ij}}{u_j}  \]
because of the assumption on the derivatives that occur in $\cL_{ij}$, which excludes mixed derivatives with respect to $t_i$ and $t_j$.

As Hamilton function we can take
\[ H_{ij} = p^{[i]} u_j - \cL_{ij} . \]
Its formal integral $\int H_{ij} \, \d x \wedge \d t_i$ does not depend on any $t_j$-derivatives. Since we are working with 2-dimensional integrals, we should take a $(2,2)$-form as symplectic form. In analogy to Equation \eqref{symp} we take
\[ \omega_i = - \delta p^{[i]} \wedge \delta u \wedge \d x \wedge \d t_i . \]
A Hamiltonian vector field $X =  \chi \der{}{u}$ satisfies
\[ \int \iota_X \omega_i = \int \delta H_{ij} \wedge \d x \wedge \d t_i \]
hence
\[ \cE_{p^{[i]}} \chi = \var{1i}{H_{ij}}{u} , \]
where $\cE_{p^{[i]}}$ is the differential operator
\[ \cE_{p^{[i]}} = \sum_{k = 0}^\infty \left( (-1)^k \partial_x^k \der{p^{[i]}}{u_{x^k}} - \der{p^{[i]}}{u_{x^k}} \partial_x^k \right) \]
The corresponding (nonlocal) Poisson bracket is
\[ \left\{ \tint f \, \d x \wedge \d t_i  , \tint g \, \d x \wedge \d t_i \right\}_i =  -\int \var{1i}{f}{u} \, \cE_{p^{[i]}}^{-1} \, \var{1i}{g}{u}  \, \d x \wedge \d t_i  . \]
Note that $H$ is not skew-symmetric, $H_{ij} \neq H_{ji}$.

The space of functionals $ \cV \big/ \left( \partial_1 \! \cV + \partial_i \! \cV \right)$, on which the Poisson bracket $\{ \cdot , \cdot \}_i$ is defined, depends on $i$ and is different from the space of functionals for the bracket $\{ \cdot , \cdot \}$ from Equation \eqref{poisson2}. Hence no pair of these brackets are compatible with each other in the sense of a bi-Hamiltonian system.

As before, we can relate Poisson brackets between the Hamilton functionals to coefficients of $\d \cL$.
\begin{prop}
	Assume that for all $i,j> 1$, $\cL_{ij}$ does not depend on any second or higher derivatives with respect to $t_i$ and $t_j$. On solutions of the Euler-Lagrange equations there holds that, for $i,j,k > 1$,
	\begin{equation}\label{triple2}
		\begin{split}
			\int \left( \frac{\d \cL_{ij}}{\d t_k} - \frac{\d \cL_{ik}}{\d t_j}  \right) \d x \wedge \d t_i
			&= \int \left( p^{[i]}_j u_k - p^{[i]}_k u_j  \right) \d x \wedge \d t_i \\
			&= \left\{ \tint H_{ij} \, \d x \wedge \d t_i , \tint H_{ik} \, \d x \wedge \d t_i  \right\}_i .
		\end{split}
	\end{equation}
\end{prop}
\begin{proof}
	Analogous to the proof of Proposition \ref{prop-triple1}.
\end{proof}

\begin{example}
	For the potential KdV equation (see Example \ref{ex-kdv}) we have
	\[ p^{[2]} = \var{1}{\cL_{23}}{u_3} = \frac{3}{2} u_{111} + \frac{3}{2} u_1^2, \]
	hence
	\begin{align*}
		\cE_{p^{[2]}} &= -\partial_1 \der{p^{[i]}}{u_{1}} - \partial_1^3 \der{p^{[i]}}{u_{111}} - \der{p^{[i]}}{u_{1}} \partial_1 - \der{p^{[i]}}{u_{111}} \partial_1^3 \\
		&= - 3 \partial_1 u_1 - \frac{3}{2} \partial_1^3 - 3 u_1 \partial_1 - \frac{3}{2} \partial_1^3 \\
		&= -3 \partial_1^3 -6 u_1 \partial_1 - 3 u_{11}.
	\end{align*}
	We have
	\begin{align*}
		H_{23} &= p^{[2]} u_3 - \cL_{23} \\
		&= -3 u_{1}^{5} + \frac{15}{2} u_{1}^{2} u_{11}^{2} - 10 u_{1}^{3} u_{111} + 5 u_{1}^{3} u_{3} - \frac{7}{2} u_{11}^{2} u_{111} - 3 u_{1} u_{111}^{2} + 6 u_{1} u_{11} u_{1111} \\
		&\quad - \frac{3}{2} u_{1}^{2} u_{11111} - 10 u_{1} u_{11} u_{12} + \frac{5}{2} u_{11}^{2} u_{2} + 5 u_{1} u_{111} u_{2} + \frac{1}{2} u_{1111}^{2} - \frac{1}{2} u_{111} u_{11111} \\
		&\quad  + \frac{1}{2} u_{111} u_{112} - \frac{1}{2} u_{1} u_{113} - u_{1111} u_{12} + \frac{1}{2} u_{11} u_{13} + \frac{1}{2} u_{11111} u_{2} + u_{111} u_{3} ,
	\end{align*}
	where the terms involving $t_3$-derivatives cancel out when the Hamiltonian is integrated. Its variational derivative is
	\begin{align*}
		\var{12}{H_{23}}{u} &= 60 u_{1}^3 u_{11} + 75 u_{11}^3 + 300 u_{1} u_{11} u_{111} + 75 u_{1}^2 u_{1111} - 30 u_{1}^2 u_{12} - 30 u_{1} u_{11} u_{2} \\
		&\quad + 120 u_{111} u_{1111} + 72 u_{11} u_{11111} + 24 u_{1} u_{111111} - 30 u_{1} u_{1112} - 45 u_{11} u_{112} \\
		&\quad - 25 u_{111} u_{12} - 5 u_{1111} u_{2} + 2 u_{11111111} - 5 u_{111112} .
	\end{align*}
	On solutions this simplifies to
	\begin{align*}
		\var{12}{H_{23}}{u} &= -210 u_{1}^3 u_{11} - 195 u_{11}^3 - 690 u_{1} u_{11} u_{111} - 150 u_{1}^2 u_{1111} - 210 u_{111} u_{1111} \\
		&\quad- 123 u_{11} u_{11111} - 36 u_{1} u_{111111} - 3 u_{11111111} \\
		&= \cE_{p^{[2]}} \left( 10 u_1^3 + 5 u_{11}^2 + 10 u_1 u_{111} + u_{11111} \right) \\
		&= \cE_{p^{[2]}} u_3 .
	\end{align*}
	Hence
	\[ \frac{\d}{\d t_3} \int u \,\d x \wedge \d t_2
	= \int \cE_{p^{[2]}}^{-1} \var{12}{H_{23}}{u} \,\d x \wedge \d t_2
	= \left\{ \tint H_{23} \, \d x \wedge \d t_2 , \tint u \, \d x \wedge \d t_2 \right\}_2 . \]
\end{example}

\subsection{Comparison with the covariant approach}

In Section \ref{sec-additional} we derived Poisson brackets $\{\cdot,\cdot\}_i$, associated to each time variable $t_i$. This was somewhat cumbersome because we had a priori assigned $x = t_1$ as a distinguished variable. The recent work \cite{caudrelier2020hamiltonian} explores the relation of pluri-Lagrangian structures to covariant Hamiltonian structures. The meaning of ``covariant'' here is that all variables are on the same footing; there is no distinguished $x$ variable. More details on covariant field theory, and its connection to the distinguished-variable (or ``instantaneous'') perspective, can be found in \cite{kanatchikov1998canonical}. The main objects in the covariant Hamiltonian formulation of \cite{caudrelier2020hamiltonian} are:
\begin{itemize}
	\item A ``symplectic multiform'' $\Omega$, which can be expanded as
	\[ \Omega = \sum_j \omega_j \wedge \d t_j, \]
	where each $\omega_j$ is a vertical 2-form in the variational bicomplex.
	
	\item A ``Hamiltonian multiform'' $\mathcal{H} = \sum_{i<j} \mathtt{H}_{ij} \,\d t_i \wedge \d t_j$ which gives the equations of motion through
	\begin{equation}\label{covariant}
		\delta \mathcal{H} = \sum_j \d t_j \wedge \xi_j \lrcorner\, \Omega,
	\end{equation}
	where $\delta$ is the vertical exterior derivative in the variational bicomplex, $\xi_j$ denotes the vector field corresponding to the $t_j$-flow, and $\lrcorner$ denotes the interior product. This equation should be understood as a covariant version of the instantaneous Hamiltonian equation $\delta H = \xi \lrcorner\, \omega$. On the equations of motion there holds $\d \mathcal{H} = 0$ if and only if $\d \cL = 0$.
	
	Since the covariant Hamiltonian equation \eqref{covariant} is of a different form than the instantaneous Hamiltonian equation we use, the coefficients $\mathtt{H}_{ij}$ of the Hamiltonian multiform $\mathcal{H}$ are also different from the $H_{ij}$ we found in Sections \ref{sec-pL2Ham}--\ref{sec-additional}. Our $H_{ij}$ are instantaneous Hamiltonians where $t_1$ and $t_i$ are considered as space variables and the Legendre transformation has been applied with respect to $t_j$.
	
	\item A ``multi-time Poisson bracket'' $\{|\cdot,\cdot|\}$ which defines a pairing between functions or (a certain type of)  horizontal one-forms, defined by 
	\[ \{|F,G|\} = (-1)^r \xi_F \delta G, \]
	where $\xi_F$ is the Hamiltonian (multi-)vector field associated to $F$, and $r$ is the horizontal degree of $F$ (which is either 0 or 1). The equations of motion can be written as 
	\[ \d F =  \sum_{i<j} \{|\mathtt{H}_{ij},F|\} \,\d t_i \wedge \d t_j. \]
\end{itemize}

Single-time Poisson brackets are obtained in \cite{caudrelier2020hamiltonian} by expanding the multi-time Poisson bracket as 
\[ \left\{\left| {\textstyle \sum_j} F_j \,\d t_j, {\textstyle \sum_j} G_j \,\d t_j \right|\right\} = \sum_j \{F_j,G_j\}_j \,\d t_j \]
where 
\begin{equation}\label{covariant-pb}
	\{f,g\}_j = - \xi_f^j \lrcorner\, \delta g \qquad \text{and} \qquad \xi_f^j \lrcorner\, \omega_j = \delta f.
\end{equation}
These are fundamentally different from the Poisson brackets of Sections \ref{sec-pL2Ham}--\ref{sec-additional} because they act on different function spaces.
Equation \eqref{covariant-pb} assumes that $\delta f$ lies in the image of $\omega_j$ (considered as a map from vertical vector fields to vertical one-forms).
For example, for the potential KdV hiararchy one has $\omega_1 = \delta v \wedge \delta v_1$, hence the Poisson bracket $\{\cdot,\cdot\}_1$ can only be applied to functions of $v$ and $v_1$, not to functions depending on any higher derivatives. Similar conditions on the function space apply to the higher Poisson brackets corresponding to $\omega_j$, $j\geq2$. On the other hand, the Poisson brackets of Sections \ref{sec-pL2Ham}--\ref{sec-additional} are defined on an equivalence class of functions modulo certain derivatives, without further restrictions on the functions in this class.

In summary, the single-time Poisson brackets of \cite{caudrelier2020hamiltonian} are constructed with a certain elegance in a covariant way, but they are defined only in a restricted function space. They are different from our Poisson brackets of Section \ref{sec-pL2Ham}--\ref{sec-additional}, which have no such restrictions, but break covariance already in the definition of the function space as an equivalence class. It is not clear how to pass from one picture to the other, or if their respective benefits can be combined into a single approach.

\section{Conclusions}

We have established a connection between pluri-Lagrangian systems and integrable Hamiltonian hierarchies. In the case of ODEs, where the pluri-Lagrangian structure is a 1-form, this connection was already obtained in \cite{suris2013variational}. Our main contribution is its generalization to the case of 2-dimensional PDEs, described by Lagrangian 2-forms. Presumably, this approach extends to Lagrangian $d$-forms of any dimension $d$, but the details of this are postponed to future work.

A central property in the theory of pluri-Lagrangian systems is that the Lagrangian form is (almost) closed on solutions. We showed that closedness is equivalent to the corresponding Hamilton functions being in involution.

Although one can obtain several Poisson brackets (and corresponding Hamilton functions) from one Lagrangian 2-form, these do not form a bi-Hamiltonian structure and it is not clear if a recursion operator can be obtained from them. Hence it remains an open question to find a fully variational description of bi-Hamiltonian hierarchies.

\subsection*{Acknowledgements}

The author would like to thank Frank Nijhoff for his inspiring questions and comments, Matteo Stoppato for helpful discussions about the covariant Hamiltonian approach, Yuri Suris for his constructive criticism on early drafts of this paper, and the anonymous referees for their thoughtful comments.

The author is funded by DFG Research Fellowship VE 1211/1-1. Part of the work presented here was done at TU Berlin, supported by the SFB Transregio 109 ``Discretization in Geometry and Dynamics''.

\appendix
\section{Pluri-Lagrangian systems and the variational bicomplex}
\label{sec-appendix}

In this appendix we study the pluri-Lagrangian principle using the variational bicomplex, described in Section \ref{sec-varbi}. We provide proofs that the multi-time Euler-Lagrange equations from Section \ref{sec-plurilag} are sufficient conditions for criticality. Alternative proofs of this fact can be found in \cite{suris2016lagrangian} and \cite[Appendix A]{sleigh2020lagrangian}.

\begin{prop}
	\label{prop-delta=d}
	The field $u$ is a solution to the pluri-Lagrangian problem of a $d$-form $\cL \llbracket u \rrbracket$ if locally there exists a $(1,d-1)$-form $\Theta$ such that $\delta \cL \llbracket u \rrbracket = \d \Theta$.
\end{prop}
\begin{proof}
	Consider a field $u$ such that such a $(1,d-1)$-form $\Theta$ exists. Consider any $d$-manifold $\Gamma$ and any variation $v$ that vanishes (along with all its derivatives) on the boundary $\partial \Gamma$. Note that the horizontal exterior derivative $\d$ anti-commutes with the interior product operator $\iota_{V}$, where $V$ is the prolonged vertical vector field $V = \pr (v \partial / \partial_u)$ defined by the variation $v$. It follows that
	\[ \int_\Gamma \iota_{V} \delta \cL = -\int_\Gamma \d \left( \iota_{V} \Theta \right) = -\int_{\partial \Gamma} \iota_{V} \Theta = 0 ,\]
	hence $u$ solves the pluri-Lagrangian problem.
\end{proof}

If we are dealing with a classical Lagrangian problem from mechanics, $\cL = L(u,u_t) \,\d t$, we have $\Theta = -\der{L}{u_t} \delta u$, which is the pull back to the tangent bundle of the canonical 1-form $\sum_i p_i \,\d q_i$ on the cotangent bundle. 

Often we want the Lagrangian form to be closed when evaluated on solutions. As we saw in Theorems \ref{thm-dL0-1} and \ref{thm-dL0-2}, this implies that the corresponding Hamiltonians are in involution. We did not include this in the definition of a pluri-Lagrangian system, because our definition already implies a slightly weaker property.

\begin{prop}\label{prop-almost-closed}
	The horizontal exterior derivative $\d \cL$ of a pluri-Lagrangian form is constant on connected components of the set of critical fields for $\cL$.
\end{prop}
\begin{proof}
	Critical points satisfy locally
	\[ \delta \cL = \d \Theta
	\qquad \Rightarrow \qquad
	\d \delta  \cL = 0
	\qquad \Rightarrow \qquad
	\delta \d \cL = 0.\]
	Hence for any variation $v$ the Lie derivative of $\d \cL$ along its prolongation $V = \pr (v \partial / \partial_u)$ is $\iota_{V} \delta (\d \cL) = 0$. Therefore, if a solution $u$ can be continuously deformed into another solution $\bar{u}$, then $\d \cL \llbracket u \rrbracket = \d \cL \llbracket \bar{u} \rrbracket$.
\end{proof}

Now let us prove the sufficiency of the multi-time Euler-Lagrange equations for 1-forms and 2-forms, as given in Theorems \ref{thm-EL1} and \ref{thm-EL2}. For different approaches to the multi-time Euler-Lagrange equations, including proofs of necessity, see \cite{suris2016lagrangian} and \cite{sleigh2020lagrangian}.

\begin{proof}[Proof of sufficiency in Theorem \ref{thm-EL1}.]
	We calculate the vertical exterior derivative $\delta \cL$ of the Lagrangian 1-form, modulo the multi-time Euler-Lagrange Equations \eqref{el1} and \eqref{el2}. We have
	\begin{align*}
		\delta \cL &= \sum_{j=1}^N \sum_I \der{\cL_j}{u_I} \, \delta u_I \wedge \d t_j \\
		&= \sum_{j=1}^N \sum_I \left( \var{j}{\cL_j}{u_I} + \partial_j \var{j}{\cL_j}{u_{It_j}}  \right) \delta u_I \wedge \d t_j .
	\end{align*}
	Rearranging this sum, we find
	\[ \delta \cL
	= \sum_{j=1}^N \left[ \sum_{I \not\ni t_j} \var{j}{\cL_j}{u_I} \delta u_I \wedge \d t_j  + \sum_I \left( \var{j}{\cL_j}{u_{It_j}} \delta u_{I t_j} \wedge \d t_j + \left( \partial_j \var{j}{\cL_j}{u_{It_j}} \right) \delta u_I \wedge \d t_j \right) \right]. \]
	On solutions of Equation \eqref{el2}, we can define the generalized momenta
	\[ p^I = \var{j}{\cL_j}{u_{It_j}} . \]
	Using Equations \eqref{el1} and \eqref{el2} it follows that
	\begin{align*}
		\delta \cL 
		&= \sum_{j=1}^N \sum_I \left( p^I \delta u_{I t_j} \wedge \d t_j + \left( \partial_j p^I \right) \delta u_I \wedge \d t_j \right) 
		= -d \left( \sum_I p^I \delta u_I \right).
	\end{align*}
	This implies by Proposition \ref{prop-delta=d} that $u$ solves the pluri-Lagrangian problem. 
\end{proof}

\begin{proof}[Proof of sufficiency in Theorem \ref{thm-EL2}.]
	We calculate the vertical exterior derivative $\delta \cL$, 
	\begin{align}
		\delta \cL
		&= \sum_{i < j} \sum_I \der{\cL_{ij}}{u_I} \, \delta u_I \wedge \d t_i \wedge \d t_j \notag
		\\
		\label{4term}
		&= \sum_{i < j} \sum_I \left( \var{ij}{\cL_{ij}}{u_I} + \partial_i \var{ij}{\cL_{ij}}{u_{I t_i}} + \partial_j \var{ij}{\cL_{ij}}{u_{I t_j}} + \partial_i \partial_j \var{ij}{\cL_{ij}}{u_{I t_i t_j}} \right) \delta u_I \wedge \d t_i \wedge \d t_j
	\end{align}
	We will rearrange this sum according to the times occurring in the multi-index $I$. We have
	\begin{align*}
		&\sum_I \var{ij}{\cL_{ij}}{u_I} \,\delta u_I = \sum_{I \not\ni t_i,t_j} \var{ij}{\cL_{ij}}{u_I} \,\delta u_I  + \sum_{I \not\ni t_j} \var{ij}{\cL_{ij}}{u_{I t_i}} \,\delta u_{I t_i} \\
		&\hspace{3cm} + \sum_{I \not\ni t_i} \var{ij}{\cL_{ij}}{u_{I t_j}} \,\delta u_{I t_j} + \sum_I \var{ij}{\cL_{ij}}{u_{I t_i t_j}} \,\delta u_{I t_i t_j} , 
		\\
		&\sum_I \partial_i \var{ij}{\cL_{ij}}{u_{I t_i}} \,\delta u_I = \sum_{I \not\ni t_j} \partial_i \var{ij}{\cL_{ij}}{u_{I t_i}} \,\delta u_{I} + \sum_{I} \partial_i \var{ij}{\cL_{ij}}{u_{I t_i t_j}} \,\delta u_{I t_j} , 
		\\
		&\sum_I \partial_j \var{ij}{\cL_{ij}}{u_{I t_j}} \,\delta u_I = \sum_{I \not\ni t_i} \partial_j \var{ij}{\cL_{ij}}{u_{I t_j}} \,\delta u_{I} + \sum_{I} \partial_j \var{ij}{\cL_{ij}}{u_{I t_i t_j}} \,\delta u_{I t_i} .
	\end{align*}
	Modulo the multi-time Euler-Lagrange equations \eqref{2el1}--\eqref{2el3}, we can write these expressions as
	\begin{align*}
		&\sum_I \var{ij}{\cL_{ij}}{u_I} \,\delta u_I = \sum_{I \not\ni t_j} p_j^I \,\delta u_{I t_i} - \sum_{I \not\ni t_i} p_i^I \,\delta u_{I t_j} + 	\sum_I (n_j^I - n_i^I) \,\delta u_{I t_i t_j} , 
		\\
		&\sum_I \partial_i \var{ij}{\cL_{ij}}{u_{I t_i}} \,\delta u_I = \sum_{I \not\ni t_j} \partial_i p_j^I \,\delta u_{I} + \sum_{I} \partial_i (n_j^I - n_i^I) \,\delta u_{I t_j} , 
		\\
		&\sum_I \partial_j \var{ij}{\cL_{ij}}{u_{I t_j}} \,\delta u_I = \sum_{I \not\ni t_i} - \partial_j p_i^I \,\delta u_{I} + \sum_{I} \partial_j (n_j^I - n_i^I) \,\delta u_{I t_i} , 
		\\
		&\sum_I \partial_i \partial_j \var{ij}{\cL_{ij}}{u_{I t_i t_j}} \,\delta u_I = \sum_I \partial_i \partial_j (n_j^I - n_i^I) \delta u_I.
	\end{align*}
	where
	\begin{align*}
		& p_j^I = \var{1j}{\cL_{1j}}{u_{It_1}} \qquad \text{for } I \not\ni t_j, \\
		& n_j^I = \var{1j}{\cL_{1j}}{u_{It_1t_j}} .
	\end{align*} 
	Note that here the indices of $p$ and $n$ are labels, not derivatives. Hence on solutions to equations \eqref{2el1}--\eqref{2el3}, Equation \eqref{4term} is equivalent to
	\begin{align*}
		\delta \cL
		&=  \sum_{i < j}  \Bigg[
		\sum_{I \not\ni t_j}  \left( p_j^I \,\delta u_{I t_i} + \partial_i p_j^I \,\delta u_I \right) - \sum_{I \not\ni t_i}  \left( p_i^I \,\delta u_{I t_j} + \partial_j p_i^I \,\delta u_I \right) \\
		&\hspace{15mm} + \sum_I \Big( (n_j^I-n_i^I) \,\delta u_{I t_i t_j} + \partial_j (n_j^I-n_i^I) \,\delta u_{I t_i} \\
		&\hspace{3cm} + \partial_i (n_j^I-n_i^I) \,\delta u_{I t_j} + \partial_i \partial_j (n_j^I-n_i^I) \,\delta u_I \Big) \Bigg]  \wedge \d t_i \wedge \d t_j.
	\end{align*}
	Using the anti-symmetry of the wedge product, we can write this as
	\begin{align*}
		\delta \cL &= \sum_{i,j=1}^N \Bigg[
		\sum_{I \not\ni t_j}  \left( p_j^I \, \delta u_{I t_i} + \partial_i p_j^I \,\delta u_I  \right) \\
		&\hspace{15mm} + \sum_I \Big( n_j^I \,\delta u_{I t_i t_j} + \partial_j n_j^I \,\delta u_{I t_i} + \partial_i n_j^I \,\delta u_{I t_j} + \partial_i \partial_j n_j^I \,\delta u_I  \Big) \Bigg]  \wedge \d t_i \wedge \d t_j
		\\
		&= \sum_{j=1}^N \Bigg[
		\sum_{I \not\ni t_j} - \d \left( p_j^I \, \delta u_I \wedge \d t_j \right)  + \sum_I -d \left( n_j^I \,\delta u_{I t_j} \wedge \d t_j + \partial_j n_j^I \,\delta u_{I} \wedge \d t_j \right) \Bigg] .
	\end{align*}
	It now follows by Proposition \ref{prop-delta=d} that $u$ is a critical field.
\end{proof}

\strut%\vfill
 
%\pagebreak
 
\bibliographystyle{abbrvnat_mv}
\bibliography{ham}

\begin{thebibliography}{33}
\providecommand{\natexlab}[1]{#1}
\providecommand{\url}[1]{\texttt{#1}}
\expandafter\ifx\csname urlstyle\endcsname\relax
  \providecommand{\doi}[1]{doi: #1}\else
  \providecommand{\doi}{doi: \begingroup \urlstyle{rm}\Url}\fi

\bibitem[Anderson(1992)]{anderson1992introduction}
Anderson~I.~M.
\newblock \textit{Introduction to the variational bicomplex}.
\newblock In Gotay~M., Marsden~J. \& Moncrief~V., editors, \emph{Mathematical
  Aspects of Classical Field Theory}, pages \mbox{51--73}. AMS, 1992.

\bibitem[Bergvelt and De~Kerf(1985)]{bergvelt1985poisson}
Bergvelt~M.~J. \& De~Kerf~E.~A.
\newblock \textit{Poisson brackets for {L}agrangians linear in the velocity}.
\newblock \href{http://dx.doi.org/10.1007/BF00704581}{Letters in Mathematical
  Physics, 10\,:\,\mbox{13--19}}, 1985.

\bibitem[Bobenko and Suris(2010)]{bobenko2010lagrangian}
Bobenko~A.~I. \& Suris~{\relax Yu}.~B.
\newblock \textit{On the {L}agrangian structure of integrable quad-equations}.
\newblock \href{http://dx.doi.org/10.1007/s11005-010-0381-9}{Letters in
  Mathematical Physics, 92\,:\,\mbox{17--31}}, 2010.

\bibitem[Bobenko and Suris(2015)]{bobenko2015discrete}
Bobenko~A.~I. \& Suris~{\relax Yu}.~B.
\newblock \textit{Discrete pluriharmonic functions as solutions of linear
  pluri-{L}agrangian systems}.
\newblock \href{http://dx.doi.org/10.1007/s00220-014-2240-5}{Communications in
  Mathematical Physics, 336\,:\,\mbox{199--215}}, 2015.

\bibitem[Boll et~al.(2014)Boll, Petrera, and Suris]{boll2014integrability}
Boll~R., Petrera~M. \& Suris~{\relax Yu}.~B.
\newblock \textit{What is integrability of discrete variational systems?}
\newblock \href{http://dx.doi.org/10.1098/rspa.2013.0550}{Proceedings of the
  Royal Society A, 470\,:\,\mbox{20130550}}, 2014.

\bibitem[Caudrelier and Stoppato(2020)]{caudrelier2020hamiltonian}
Caudrelier~V. \& Stoppato~M.
\newblock \textit{Hamiltonian multiform description of an integrable
  hierarchy}.
\newblock \href{http://dx.doi.org/10.1063/5.0012153}{Journal of Mathematical
  Physics, 61\,:\,\mbox{123506}}, 2020.

\bibitem[De~Sole and Kac(2013)]{desole2013nonlocal}
De~Sole~A. \& Kac~V.~G.
\newblock \textit{Non-local {P}oisson structures and applications to the theory
  of integrable systems}.
\newblock \href{http://dx.doi.org/10.1007/s11537-013-1306-z}{Japanese Journal
  of Mathematics, 8\,:\,\mbox{233--347}}, 2013.

\bibitem[Dickey(2003)]{dickey2003soliton}
Dickey~L.~A.
\newblock \textit{Soliton Equations and Hamiltonian Systems}.
\newblock World Scientific, 2nd edition, 2003.

\bibitem[Dorfman(1987)]{dorfman1987dirac}
Dorfman~I.~Y.
\newblock \textit{Dirac structures of integrable evolution equations}.
\newblock \href{http://dx.doi.org/10.1016/0375-9601(87)90201-5}{Physics Letters
  A, 125\,:\,\mbox{240--246}}, 1987.

\bibitem[Gardner(1971)]{gardner1971korteweg}
Gardner~C.~S.
\newblock \textit{{K}orteweg-de {V}ries equation and generalizations. {IV}. the
  {K}orteweg-de {V}ries equation as a {H}amiltonian system}.
\newblock \href{http://dx.doi.org/10.1063/1.1665772}{Journal of Mathematical
  Physics, 12\,:\,\mbox{1548--1551}}, 1971.

\bibitem[Gotay(1988)]{gotay1988multisymplectic}
Gotay~M.~J.
\newblock \textit{A multisymplectic approach to the {KdV} equation}.
\newblock In
  \href{http://dx.doi.org/10.1007/978-94-015-7809-7_15}{\emph{Differential
  Geometrical Methods in Theoretical Physics}, pages \mbox{295--305}}.
  Springer, 1988.

\bibitem[Hietarinta et~al.(2016)Hietarinta, Joshi, and
  Nijhoff]{hietarinta2016discrete}
Hietarinta~J., Joshi~N. \& Nijhoff~F.~W.
\newblock \textit{Discrete Systems and Integrability}.
\newblock \href{http://dx.doi.org/10.1017/CBO9781107337411}{Cambridge
  University Press}, Cambridge, 2016.

\bibitem[Kanatchikov(1998)]{kanatchikov1998canonical}
Kanatchikov~I.~V.
\newblock \textit{Canonical structure of classical field theory in the
  polymomentum phase space}.
\newblock \href{http://dx.doi.org/10.1016/S0034-4877(98)80182-1}{Reports on
  Mathematical Physics, 41\,:\,\mbox{49--90}}, 1998.

\bibitem[Lobb and Nijhoff(2009)]{lobb2009lagrangian}
Lobb~S. \& Nijhoff~F.
\newblock \textit{Lagrangian multiforms and multidimensional consistency}.
\newblock \href{http://dx.doi.org/10.1088/1751-8113/42/45/454013}{Journal of
  Physics A: Mathematical and Theoretical, 42\,:\,\mbox{454013}}, 2009.

\bibitem[Macfarlane(1982)]{macfarlane1982equations}
Macfarlane~A.~J.
\newblock \textit{Equations of {K}orteweg-de {V}ries type {I}: {L}agrangian and
  {H}amiltonian formalism}.
\newblock Technical Report TH-3289, CERN.
\newblock \url{http://cds.cern.ch/record/137678}, 1982.

\bibitem[Mokhov(1998)]{mokhov1998symplectic}
Mokhov~O.~I.
\newblock \textit{Symplectic and poisson structures on loop spaces of smooth
  manifolds, and integrable systems}.
\newblock \href{http://dx.doi.org/10.1070/RM1998v053n03ABEH000019}{Russian
  Mathematical Surveys, 53\,:\,\mbox{515}}, 1998.

\bibitem[Olver(2000)]{olver2000applications}
Olver~P.~J.
\newblock \textit{Applications of {L}ie Groups to Differential Equations}.
\newblock Volume 107 of \emph{Graduate Texts in Mathematics}. Springer, 2nd
  edition, 2000.

\bibitem[Petrera and Suris(2017)]{petrera2017variational}
Petrera~M. \& Suris~{\relax Yu}.~B.
\newblock \textit{Variational symmetries and pluri-{L}agrangian systems in
  classical mechanics}.
\newblock \href{http://dx.doi.org/10.1080/14029251.2017.1418058}{Journal of
  Nonlinear Mathematical Physics, 24 (Sup.~1)\,:\,\mbox{121--145}}, 2017.

\bibitem[Petrera and Vermeeren(2020)]{petrera2019variational}
Petrera~M. \& Vermeeren~M.
\newblock \textit{Variational symmetries and pluri-{L}agrangian structures for
  integrable hierarchies of {PDE}s}.
\newblock \href{http://dx.doi.org/10.1007/s40879-020-00436-7}{European Journal
  of Mathematics}, 2020.

\bibitem[Rowley and Marsden(2002)]{rowley2002variational}
Rowley~C.~W. \& Marsden~J.~E.
\newblock \textit{Variational integrators for degenerate {L}agrangians, with
  application to point vortices}.
\newblock In
  \href{http://dx.doi.org/10.1109/CDC.2002.1184735}{\emph{Proceedings of the
  41st IEEE Conference on Decision and Control, 2002.}, pages
  \mbox{1521--1527}}. IEEE, 2002.

\bibitem[Sleigh et~al.(2019{\natexlab{a}})Sleigh, Nijhoff, and
  Caudrelier]{sleigh2019lax}
Sleigh~D., Nijhoff~F. \& Caudrelier~V.
\newblock \textit{A variational approach to {L}ax representations}.
\newblock \href{http://dx.doi.org/10.1016/j.geomphys.2019.03.015}{Journal of
  Geometry and Physics, 142\,:\,\mbox{66--79}}, 2019{\natexlab{a}}.

\bibitem[Sleigh et~al.(2019{\natexlab{b}})Sleigh, Nijhoff, and
  Caudrelier]{sleigh2019variational}
Sleigh~D., Nijhoff~F. \& Caudrelier~V.
\newblock \textit{Variational symmetries and {L}agrangian multiforms}.
\newblock \href{http://dx.doi.org/10.1007/s11005-019-01240-5}{Letters in
  Mathematical Physics\,:\,\mbox{1--22}}, 2019{\natexlab{b}}.

\bibitem[Sleigh et~al.(2020)Sleigh, Nijhoff, and
  Caudrelier]{sleigh2020lagrangian}
Sleigh~D., Nijhoff~F. \& Caudrelier~V.
\newblock \textit{Lagrangian multiforms for {K}adomtsev-{P}etviashvili ({KP})
  and the {G}elfand-{D}ickey hierarchy}.
\newblock \href{https://arxiv.org/abs/2011.04543}{arXiv:2011.04543}, 2020.

\bibitem[Sridhar and Suris(2019)]{sridhar2019commutativity}
Sridhar~A. \& Suris~{\relax Yu}.~B.
\newblock \textit{Commutativity in {L}agrangian and {H}amiltonian mechanics}.
\newblock \href{http://dx.doi.org/10.1016/j.geomphys.2018.09.019}{Journal of
  Geometry and Physics, 137\,:\,\mbox{154--161}}, 2019.

\bibitem[Suris(2003)]{suris2003problem}
Suris~{\relax Yu}.~B.
\newblock \textit{The Problem of Integrable Discretization: Hamiltonian
  Approach}.
\newblock Birkh{\"a}user, 2003.

\bibitem[Suris(2013)]{suris2013variational}
Suris~{\relax Yu}.~B.
\newblock \textit{Variational formulation of commuting {H}amiltonian flows:
  Multi-time {L}agrangian 1-forms}.
\newblock \href{http://dx.doi.org/10.3934/jgm.2013.5.365}{Journal of Geometric
  Mechanics, 5\,:\,\mbox{365--379}}, 2013.

\bibitem[Suris(2016)]{suris2016variational}
Suris~{\relax Yu}.~B.
\newblock \textit{Variational symmetries and pluri-{L}agrangian systems}.
\newblock In
  \href{http://dx.doi.org/10.1142/9789814699877_0013}{\emph{Dynamical Systems,
  Number Theory and Applications: A Festschrift in Honor of Armin Leutbecher's
  80th Birthday}, pages \mbox{255--266}}. World Scientific, 2016.

\bibitem[Suris and Vermeeren(2016)]{suris2016lagrangian}
Suris~{\relax Yu}.~B. \& Vermeeren~M.
\newblock \textit{On the {L}agrangian structure of integrable hierarchies}.
\newblock In
  \href{http://dx.doi.org/10.1007/978-3-662-50447-5_11}{Bobenko~A.~I., editor,
  \emph{Advances in Discrete Differential Geometry}, pages \mbox{347--378}}.
  Springer, 2016.

\bibitem[Vermeeren(2019{\natexlab{a}})]{vermeeren2019continuum}
Vermeeren~M.
\newblock \textit{Continuum limits of pluri-{L}agrangian systems}.
\newblock \href{http://dx.doi.org/10.1093/integr/xyy020}{Journal of Integrable
  Systems, 4\,:\,\mbox{xyy020}}, 2019{\natexlab{a}}.

\bibitem[Vermeeren(2019{\natexlab{b}})]{vermeeren2019variational}
Vermeeren~M.
\newblock \textit{A variational perspective on continuum limits of {ABS} and
  lattice {GD} equations}.
\newblock \href{http://dx.doi.org/10.3842/SIGMA.2019.044}{SIGMA,
  15\,:\,\mbox{044}}, 2019{\natexlab{b}}.

\bibitem[Wilson(1988)]{wilson1988quasi}
Wilson~G.
\newblock \textit{On the quasi-hamiltonian formalism of the {KdV} equation}.
\newblock \href{http://dx.doi.org/10.1016/0375-9601(88)90510-5}{Physics Letters
  A, 132\,:\,\mbox{445--450}}, 1988.

\bibitem[Xenitidis et~al.(2011)Xenitidis, Nijhoff, and
  Lobb]{xenitidis2011lagrangian}
Xenitidis~P., Nijhoff~F. \& Lobb~S.
\newblock \textit{On the {L}agrangian formulation of multidimensionally
  consistent systems}.
\newblock \href{http://dx.doi.org/10.1098/rspa.2011.0124}{Proceedings of the
  Royal Society A, 467\,:\,\mbox{3295--3317}}, 2011.

\bibitem[Yoo-Kong et~al.(2011)Yoo-Kong, Lobb, and Nijhoff]{yoo2011discrete}
Yoo-Kong~S., Lobb~S. \& Nijhoff~F.
\newblock \textit{Discrete-time {C}alogero-{M}oser system and {L}agrangian
  1-form structure}.
\newblock \href{http://dx.doi.org/10.1088/1751-8113/44/36/365203}{Journal of
  Physics A: Mathematical and Theoretical, 44\,:\,\mbox{365203}}, 2011.

\end{thebibliography}

\label{lastpage}
\end{document}